 \documentclass[conference, onecolumn,pagenumbers, 11pt]{IEEEtran}

\IEEEoverridecommandlockouts
\usepackage{graphicx}
\usepackage{cite}
\usepackage{color}
\usepackage{amsfonts}
\usepackage{amsthm}
\usepackage{url}
\usepackage{amsmath}
\usepackage{amssymb}
\usepackage{enumerate}
\usepackage{subfig}
\usepackage{fixltx2e}
\usepackage{dsfont}
\usepackage{pseudocode}
\usepackage{mathrsfs}

\newtheorem{theorem}{Theorem}
\newtheorem{lemma}{Lemma}
\newtheorem{remark}{Remark}

\newtheorem{corollary}{Corollary}
\newtheorem{conjecture}{Conjecture}

\newtheorem{proposition}{Proposition}

\newcommand{\EE}{\mathbb{E}}
\newcommand{\PP}{\mathbb{P}}

\newcommand{\xh}[1]{{X}_{#1}}
\newcommand{\yh}[1]{{Y}_{#1}}

\newcommand{\etal}{\emph{et al.}}

\newcommand{\vecY}{\mathbf{Y}}
\newcommand{\vecX}{\mathbf{X}}
\newcommand{\vecZ}{\mathbf{Z}}
\newcommand{\vecW}{\mathbf{W}}

\newcommand{\Var}{\operatorname{Var}}

\usepackage{eufrak}
\graphicspath{{figs/}}

\begin{document}

\title{Strengthening the Entropy Power Inequality %
\thanks{This work was supported by   NSF Grants CCF-1528132 and CCF-0939370 (Center for Science of Information). \newline Email: courtade@berkeley.edu}
}

\author{
\IEEEauthorblockN{Thomas~A.~Courtade}
\IEEEauthorblockA{Department of Electrical Engineering and Computer Sciences\\University of California, Berkeley
              }
}
\maketitle

\thispagestyle{plain}
\pagestyle{plain}

\begin{abstract}
We tighten the Entropy Power Inequality (EPI) when one of the random summands is Gaussian.  Our strengthening is closely connected to the concept of strong data processing for Gaussian channels and generalizes the (vector extension of) Costa's EPI.  This leads to a new reverse entropy power inequality and, as a corollary, sharpens Stam's inequality relating entropy power and Fisher information. 
Applications to network information theory are given, including a short self-contained proof of the rate region for the two-encoder quadratic Gaussian source coding problem.  

Our argument is based  on weak convergence and a technique employed by Geng and Nair for establishing Gaussian optimality via rotational-invariance, which traces its roots to a  `doubling trick' that has been successfully used in the study of functional inequalities.
\end{abstract}


\newcommand{\snr}{\mathsf{snr}}
\newcommand{\C}{\mathfrak{C}}

\section{Introduction and Main Result}

 For a random variable $X$ with density $f$,   the differential entropy of $X$ is defined by
\begin{align}
h(X) = -\int f(x) \log f(x) dx. \label{diffEntropy}
\end{align} 
Similarly, $h(\vecX)$ is defined to be the differential entropy of a random vector $\vecX\in \mathbb{R}^n$ having density on $\mathbb{R}^n$.  The celebrated Entropy Power Inequality (EPI) put forth by Shannon \cite{bib:Shannon48} and  rigorously established by Stam \cite{stam1959some} and Blachman \cite{bib:Blachman} asserts that for $X,W$ independent
\begin{align}
2^{2h(X+W)}\geq 2^{2h(X)}+2^{2h(W)}. \label{shannonEPI}
\end{align}
Under the assumption that $W$ is Gaussian, we prove the following strengthening of \eqref{shannonEPI}:  
\begin{theorem}\label{thm:scalarEIPI}
Let $X\sim P_X$, and let $W\sim N(0,\sigma^2)$ be independent of $X$.  For any $V$ satisfying $X \to (X+W) \to V$, 
\begin{align}
2^{2(h(X+W) - I(X;V))} \geq  2^{2(h(X) - I(X+W;V))} + 2^{2h(W)}.\label{EPItoProve}
\end{align}
\end{theorem}

The notation $X \to (X+W) \to V$ in Theorem \ref{thm:scalarEIPI} follows the usual convention, indicating that the random variables $X$, $X+W$ and $V$ form a Markov chain, in that order.  Throughout, we write $X\to Y \to V |Q$ to denote random variables $X,Y,V,Q$ with joint distribution factoring as $P_{XYVQ} = P_{XQ}P_{Y|XQ}P_{V|YQ}$.  That is, $X\to Y\to V$ form a Markov chain conditioned on $Q$.

When the  integral \eqref{diffEntropy} does not exist, or if $X$ does not have density, then we adopt the convention that $h(X) = -\infty$.    In this case, the inequality \eqref{EPItoProve} is a trivial consequence of the data processing inequality.  %
 So, as with the classical EPI, Theorem \ref{thm:scalarEIPI} is only informative when $X$ has density and $h(X)$ exists.

A conditional version of the EPI is often useful in applications.  Theorem \ref{thm:scalarEIPI} easily generalizes along these lines.  Indeed, due to joint convexity of $\log(2^{x}+2^{y})$ in $x,y$, we  obtain the following corollary of Theorem \ref{thm:scalarEIPI}:
\begin{corollary}
Suppose $X,W$ are conditionally independent given $Q$, and moreover that $W$ is conditionally Gaussian given $Q$.  Then, for any $V$ satisfying $X \to (X+W) \to V|Q$,  
\begin{align}
2^{2(h(X+W|Q) - I(X;V|Q))} \geq  2^{2(h(X|Q) - I(X+W;V|Q))} + 2^{2h(W|Q)}.\label{CondlEPItoProve}
\end{align}
\end{corollary}
It is interesting to note that the conditional version of the classical EPI assumes a form symmetric to \eqref{EPItoProve}.  In particular, for $Q \to X \to (X+W)$, it holds that 
\begin{align}
2^{2(h(X+W) - I(X+W;Q))} \geq  2^{2(h(X) - I(X;Q))} + 2^{2h(W)}. \label{classicalCondEPI}
\end{align}
Note that the mutual informations  in the exponents on the LHS and RHS of \eqref{EPItoProve} and \eqref{classicalCondEPI} respectively correspond to the smaller and larger mutual informations in the corresponding data processing inequalities $I(X;V)\leq I(X+W;V)$  and $I(X+W;Q)\leq I(X;Q)$.

As one would expect, Theorem \ref{thm:scalarEIPI} also admits a vector generalization, which may be regarded as our main result:
\begin{theorem} \label{thm:vectorEIPI}
Suppose $\vecX,\vecW$ are $n$-dimensional random vectors that are conditionally independent given $Q$, and moreover that $\vecW$ is conditionally Gaussian given $Q$.  Then, for any $V$ satisfying $\vecX \to (\vecX+\vecW) \to V|Q$,  
\begin{align}
2^{\frac{2}{n}(h(\vecX+\vecW|Q) - I(\vecX;V|Q))} \geq  2^{\frac{2}{n}(h(\vecX|Q) - I(\vecX+\vecW;V|Q))} + 2^{\frac{2}{n}h(\vecW|Q)}.\label{ineqToproveVec}
\end{align}
\end{theorem}
In the following section, we will see that the strengthening of the classical EPI afforded by Theorem \ref{thm:vectorEIPI} generalizes Costa's EPI \cite{bib:Costa} (and the vector generalization \cite{bib:GenCosta}), which has found applications ranging from interference channels to secrecy capacity (e.g., \cite{costaInterference,polyanskiy2015wasserstein,bib:GenCosta,Bagherikaram}).  It also leads to a new reverse EPI, which can be applied to improve Stam's inequality or, equivalently, the Gaussian logarithmic Sobolev inequality.  Moreover, we will see that Theorem \ref{thm:vectorEIPI}  leads to a very brief proof of the converse for the rate region of the quadratic Gaussian two-encoder source-coding problem \cite{bib:Wagner, bib:Oohama1997}.  Applications to one-sided interference channels and strong data processing inequalities are also given. 

We remark that the restriction of $\vecW$ to be conditionally Gaussian in Theorem  \ref{thm:vectorEIPI} should  not be a severe limitation in practice.  Indeed, in applications of the EPI, it is typically the case that one of the variables is Gaussian.  As noted by Rioul \cite{Rioul}, examples   include   the scalar Gaussian broadcast channel problem \cite{Bergmans} and its generalization to the multiple-input multiple-output case \cite{WSS, Cioffi}; the secrecy capacity of the Gaussian wiretap channel \cite{wiretap} and its multiple access extension \cite{TekenYener}; determination of the corner points for the scalar Gaussian interference channel problem \cite{costaInterference,polyanskiy2015wasserstein};  the scalar Gaussian source multiple-description problem \cite{Ozarow}; and characterization of the rate-distortion regions for several multiterminal Gaussian source coding schemes \cite{bib:Oohama1997, bib:Oohama2005, bib:Prabhakaran}.  It is tempting to conjecture that \eqref{ineqToproveVec} holds when the distribution of $\vecW$ is unconstrained, however we suspect this is not true (but no counterexample was immediately apparent).

\section{Applications}

\subsection{Generalized Costa's   Entropy Power Inequality}
Costa's EPI \cite{bib:Costa} states that, for independent $n$-dimensional random vectors $\vecX \sim P_{\vecX}$ and $\vecW \sim N(0,\Sigma )$, 
\begin{align}
2^{\frac{2}{n}h (\vecX +\alpha \vecW)  } &\geq (1-\alpha^2) 2^{\frac{2}{n} h(\vecX) } +\alpha^2 2^{\frac{2}{n}h(\vecX+ \vecW)} \mbox{~~~~for $|\alpha|\leq 1$.} \label{eq:costaEPI_cos}
\end{align}
This result was generalized to a vector setting by Liu \etal\, using perturbation and I-MMSE arguments \cite{bib:GenCosta}.  We demonstrate below that this generalization follows  as an easy corollary to Theorem \ref{thm:vectorEIPI} by taking $V$ equal to $\vecX$ contaminated by additive Gaussian noise.   In this sense, Theorem \ref{thm:vectorEIPI} may be interpreted as a further generalization of Costa's EPI, where the additive noise is no longer restricted to be Gaussian. 

\begin{theorem} \label{thm:GenCosta}\cite{bib:GenCosta}
Let $\vecX \sim P_{\vecX}$ and $\vecW \sim N(0,\Sigma )$ be independent, $n$-dimensional random vectors.  For a positive semidefinite matrix  $A \preceq I$, 
\begin{align}
2^{\frac{2}{n}h (\vecX + A^{1/2} \vecW)  } &\geq |I-A|^{1/n} 2^{\frac{2}{n} h(\vecX) } + |A|^{1/n} 2^{\frac{2}{n}h(\vecX+ \vecW)}.\label{eq:costaEPI}
\end{align}
\end{theorem}
\begin{proof}
Let $\vecW_1,\vecW_2$ denote two independent copies of $\vecW$, and put $\vecY = \vecX + A^{1/2} \vecW_1$ and $V = \vecY + (I-A)^{1/2} \vecW_2$.    Note that $V = \vecX + \vecW$ in distribution so that $I(\vecX;V) = h(\vecX+\vecW) - h(\vecW)$.  Similarly, $I(\vecY;V) = h(\vecX+\vecW) - h((I-A)^{1/2} \vecW)$.  
Now, \eqref{eq:costaEPI} follows from Theorem \ref{thm:vectorEIPI} since 
\begin{align}
2^{\frac{2}{n}(h (\vecX + A^{1/2} \vecW) - h(\vecX+\vecW) + h(\vecW))} &=
2^{\frac{2}{n}(h (\vecY) - I(\vecX ;V))} \\
&\geq   2^{\frac{2}{n}(h(\vecX) - I(\vecY;V))} + 2^{\frac{2}{n}h( A^{1/2} \vecW_1 )}\\
&=2^{\frac{2}{n}(h(\vecX) - h(\vecX+\vecW) + h((I-A)^{1/2}  \vecW))} + |A|^{1/n} 2^{\frac{2}{n}h( \vecW)}\\
&= |I-A|^{1/n} 2^{\frac{2}{n}(h(\vecX) - h(\vecX+\vecW) + h(  \vecW))} + |A|^{1/n} 2^{\frac{2}{n}h( \vecW)}.
\end{align}
Multiplying both sides by $2^{\frac{2}{n}( h(\vecX+\vecW) - h(\vecW))}$ completes the proof.
\end{proof}

Costa's EPI may be interpreted as a concavity property enjoyed by entropy powers.  The proof of Theorem \ref{thm:GenCosta} suggests a generalization of this property to non-Gaussian noise.  Indeed, we have the following, which may be viewed as a {reverse EPI}:
\begin{theorem}\label{thm:XZW}
Let $\vecX \sim P_{\vecX}, \vecZ \sim P_{\vecZ}$ and $\vecW \sim N(0,\Sigma )$ be independent, $n$-dimensional random vectors. Then
\begin{align}
2^{\frac{2}{n}\left( h (\vecX+\vecW)  + h(\vecZ+\vecW) \right)} &\geq 2^{\frac{2}{n}\left( h (\vecX)  + h(\vecZ) \right)} + 2^{\frac{2}{n}\left( h (\vecX+\vecZ+\vecW)  + h(\vecW) \right)}.
\end{align}
\end{theorem}
\begin{proof}
This is an immediate consequence of Theorem \ref{thm:vectorEIPI} by putting $V = \vecX+\vecZ+\vecW$ and rearranging exponents.
\end{proof}

We briefly remark that Madiman  observed the following inequality on submodularity of differential entropy \cite{madiman2008entropy}, which can be proved via data processing: if $X,Z,W$ are independent random variables, then
\begin{align}
2^{2 \left( h (X+W)  + h(Z+W) \right)} &\geq  2^{2\left( h (X+Z+W)  + h(W) \right)}. \label{madimanSubmodular}
\end{align}
When $W$ is Gaussian, Theorem \ref{thm:XZW} sharpens inequality \eqref{madimanSubmodular} by reducing the LHS by a factor of $2^{2\left( h (X)  + h(Z) \right)}$.

\subsection{A Reverse EPI and a Refinement of Stam's Inequality}

Theorem \ref{thm:XZW} admits several interesting corollaries which are deeply connected to the celebrated Gaussian Logarithmic Sobolev Inequality (LSI).  To start, define the entropy power $N(\vecX)$ and the Fisher Information $J(\vecX)$ of a random vector $\vecX$ with density $f$ with respect to Lebesgue measure as follows:
\begin{align}
&N(\vecX) \triangleq \frac{1}{2\pi e}2^{\frac{2}{n} h(\vecX)}  &J(\vecX) \triangleq \EE \left[\frac{\|\nabla f(\vecX) \|^2}{f(\vecX)}\right].
\end{align}
To avoid degeneracy, we assume throughout this section that entropies and Fisher informations exist and are finite. 

In exploring the similarity between the Brunn-Minkowski inequality and the EPI, Costa and Cover \cite{costa1984similarity} proved the following ``information isoperimetric inequality" for $n$-dimensional $\vecX$
\begin{align} 
 N(\vecX )   J(\vecX)\geq n. \label{StamInequality}
\end{align}
This inequality is commonly referred to as Stam's inequality,  due to the fact that he first observed it in his classic 1959 paper \cite{stam1959some} in the one-dimensional case.  In 1975, Gross  rediscovered \eqref{StamInequality} by establishing the (mathematically equivalent) LSI for the standard Gaussian measure $\gamma_n$ on $\mathbb{R}^n$ \cite{gross1975logarithmic}:  For every $h$ on $\mathbb{R}^n$ with gradient in $L^2(\gamma_n)$
\begin{align}
\int_{\mathbb{R}^n} h^2 \log h^2 d\gamma_n \leq 2 \int_{\mathbb{R}^n} |\nabla h|^2 d\gamma_n +\left( \int_{\mathbb{R}^n} h^2  d\gamma_n \right) \log \left( \int_{\mathbb{R}^n} h^2  ~d\gamma_n \right). \label{grossLSI}
\end{align}
In the same paper, Gross also proved that \eqref{grossLSI} is equivalent to the hypercontractivity of the Ornstein-Uhlenbeck semigroup \cite{nelson1973free}. It wasn't until the 1990's that Carlen \cite{carlen1991superadditivity} showed the equivalence between Stam's inequality and Gross' LSI. We refer the reader to \cite{raginsky2013concentration} for a concise proof and further historical details. 
 
Since \eqref{StamInequality} is proved using de Bruijn's identity and the special case of Shannon's EPI when one summand is Gaussian, Theorem \ref{thm:XZW} naturally leads to a sharpening of \eqref{StamInequality}.  Surprisingly, this strengthening takes the form of a  reverse EPI, which upper bounds $N(\vecX+\vecZ)$ in terms of the marginal entropies and Fisher informations.

\begin{theorem}\label{thm:reverseEPI}
If  $\vecX$ and  $\vecZ$ are independent $n$-dimensional random vectors, then 
\begin{align} 
 N(\vecX )  N(\vecZ) \left( J(\vecX)+  J(\vecZ)\right)\geq n  N (\vecX+\vecZ). \label{eq:reverseEPI}
\end{align}
\end{theorem}
\begin{proof}
We may assume $J(\vecX) <\infty$ and $J(\vecZ)<\infty$, else there is nothing to prove.    To begin, let   $\mathbf{G}\sim N(0,I)$ be independent of $\vecX,\vecZ$ and  recall de~Bruijn's identity \cite{stam1959some}: $\frac{d}{dt}h(\vecX + \sqrt{t} \mathbf{G})  = \frac{1} {2\ln 2}J(\vecX+ \sqrt{t} \mathbf{G})$.  In particular, we have \begin{align}
\frac{d}{dt}N(\vecX + \sqrt{t} \mathbf{G}) \Big|_{t=0}= \frac{1}{n}N(\vecX) J(\vecX).\label{deBruijnNJ}
\end{align}
Identifying $\vecW = \sqrt{t} \mathbf{G}$ in Theorem \ref{thm:XZW} and rearranging, we find
\begin{align}
\frac{ N(\vecX+\sqrt{t} \mathbf{G})  N(\vecZ+\sqrt{t} \mathbf{G}) -N (\vecX) N(\vecZ)  }{ t } &\geq N (\vecX+\vecZ+\sqrt{t} \mathbf{G})  \geq N (\vecX+\vecZ).
\end{align}
Letting $t\to0$ and applying \eqref{deBruijnNJ} proves the claim.  
\end{proof}

It is straightforward to recover Stam's inequality from Theorem \ref{thm:reverseEPI}.  Indeed,  let $\vecZ\sim N(0,\sigma^2 I)$ with variance chosen such that  $N(\vecZ )  = N(\vecX)$, then \eqref{eq:reverseEPI} reduces to 
\begin{align}
 N(\vecX ) J(\vecX)+   N(\vecZ )J(\vecZ) \geq n  \frac{N (\vecX+\vecZ)}{N(\vecX)} \geq n  \frac{2 N (\vecX)}{N(\vecX)} = 2n,
\end{align}
where the second inequality follows from the EPI. Since $N(\vecZ )J(\vecZ) = n$, \eqref{StamInequality} follows.

Stated another way, \eqref{StamInequality} reads $\frac{1}{n}J(\vecX) \geq \frac{1}{N(\vecX)}$.  Using the EPI, we may  sandwich the (appropriately normalized) entropy power of the sum $\vecX+\vecZ$ according to
\begin{align}
 \frac{1}{n}J(\vecX)+  \frac{1}{n}J(\vecZ)    \geq \frac{N (\vecX+\vecZ)}{ N(\vecX )  N(\vecZ) } \geq \frac{1}{ N(\vecX )   }+\frac{1}{N(\vecZ) },
\end{align}
which is met with equality throughout if $\vecX$ and $\vecZ$ are Gaussian with proportional covariance matrices.   

Next, let $\vecX',\vecX$ be independent and identically distributed with finite entropy, and define the \emph{doubling constant} of $\vecX$ (cf. \cite{kontoyiannis2014sumset}), denoted by $\mathsf{d}(\vecX)$, as
\begin{align}
\mathsf{d}(\vecX) \triangleq  \frac{N \left(\frac{\vecX+\vecX'}{\sqrt{2}}\right)}{ N(\vecX )  }.
\end{align}
We remark that the doubling constant and its relationship to other functionals is discussed in \cite{kontoyiannis2014sumset} for the one-dimensional  setting, and in \cite{madiman2015ruzsa} for general  dimension. 

By letting $\vecZ$ and $\vecX$ be independent and identically distributed, Theorem \ref{thm:reverseEPI} 
yields the following  inequality, which expresses the deficit in \eqref{StamInequality} in terms of the doubling constant $\mathsf{d}(\vecX)$:
\begin{corollary}\label{cor:deficitBound}
For any $n$-dimensional random vector $\vecX$, 
\begin{align}
N(\vecX)J(\vecX)    \geq n\,  \mathsf{d}(\vecX). \label{corEqn:deficitBound}
\end{align}
\end{corollary}
Recalling the  conditions for equality in the EPI,  $\mathsf{d}(\vecX)\geq 1$ with equality if and only if $\vecX$ is Gaussian.  Therefore, Corollary \ref{cor:deficitBound} represents a strict strengthening of \eqref{StamInequality}.   Since \eqref{StamInequality} is equivalent to the Gaussian LSI, Corollary \ref{cor:deficitBound}  provides a bound on the deficit in the LSI.  Such bounds have been of recent interest \cite{bobkov2014bounds, fathi2014quantitative, dolbeault2015stability}, and are interpreted  as a stability estimate for the LSI.  

Let $\vecX$ be a random vector having density $f$ with respect to $\gamma_n$.  Then the LSI \eqref{grossLSI} may be written as\footnote{In fact, this is completely equivalent to Gross' formulation \eqref{grossLSI}. } 
\begin{align}
\int_{\mathbb{R}^n} f \log f ~d\gamma_n \leq \frac{1}{2} \int_{\mathbb{R}^n} \frac{|\nabla f|^2}{f} d\gamma_n. \label{fLSI}
\end{align}
Under the assumptions that $\vecX$ is zero-mean and  satisfies the Poincar\'{e} inequality
\begin{align}
\zeta ~\EE \left[s^2(\vecX) \right]\leq  \EE\left[ |\nabla s(\vecX)|^2 \right] \label{PoincareInequality}
\end{align}
for every smooth $s :\mathbb{R}^n\to \mathbb{R}$ such that $\EE[s(\vecX) ] =0$,  the LSI \eqref{fLSI} may be improved to 
\begin{align}
\int_{\mathbb{R}^n} f \log f ~d\gamma_n \leq \frac{c(\zeta)}{2} \int_{\mathbb{R}^n} \frac{|\nabla f|^2}{f} d\gamma_n, \label{stableEst}
\end{align}
where $c(\zeta) <1$ for  `spectral gap' $\zeta >0$  \cite{fathi2014quantitative}. In fact, by using Corollary \ref{cor:deficitBound} and the self-strengthening argument of \cite{bobkov2014bounds}, the constant $c(\zeta)$ established in \cite{fathi2014quantitative} may be improved by incorporating $\mathsf{d}(\vecX)$.  

When $n=1$, $\Var(X)=1$ and $h(X)>-\infty$, Ball, Barthe and Naor \cite{BBN03} showed that the Poincar\'{e} inequality \eqref{PoincareInequality} implies 
\begin{align}
\mathsf{d}(X) \geq \left( N(X) \right)^{-\frac{\zeta}{2+2\zeta}}.\label{spectralGap}
\end{align}
  Since $N(X)\leq 1$ due to $\Var(X)=1$, we obtain a sharpening of Stam's inequality:
\begin{corollary} \label{cor:poincareSharpening}
Let $X$ be a zero-mean  random variable with $\Var(X)=1$ and finite entropy.  If $X$ satisfies the Poincar\'{e} inequality \eqref{PoincareInequality}, then
\begin{align}
 \left( N(X) \right)^{{1+\frac{3}{2}\zeta} }  \left( J(X)\right)^{{1+\zeta}}  \geq 1.
\end{align}
\end{corollary}

On account of \cite{BBN03}, a doubling constant $\mathsf{d}(X)>1$ is a weaker assumption than presence of a spectral gap.  Therefore, inequality \eqref{corEqn:deficitBound} may be viewed as an improvement on the stability estimate \eqref{stableEst} in the sense that a less restrictive hypothesis is required. 

In closing, we remark that the inequality \eqref{spectralGap} also holds for $\mathbb{R}^n$ (with $2+2\zeta$ replaced by $4+4\zeta$), provided the the density of $\vecX$ is log-concave \cite{ball2012entropy}. Thus, Corollary \ref{cor:poincareSharpening} can be modified accordingly.

\subsection{Converse for the Two-Encoder Quadratic Gaussian Source Coding Problem}
Characterizing the rate region for the two-encoder quadratic Gaussian source coding problem was a longstanding open problem in the field of network information theory until 
its ultimate resolution by Wagner \etal\,  in their landmark paper  \cite{bib:Wagner}, which established that a separation-based scheme \cite{bib:BergerLongo1977,  bib:Tung1978} was optimal.  Wagner \etal's work built upon Oohama's earlier solution  to the one-helper problem  \cite{bib:Oohama1997} and the independent solutions to the Gaussian CEO problem due to Prabhakaran, Tse and Ramachandran \cite{bib:Prabhakaran} and  Oohama \cite{bib:Oohama2005}  (see  \cite{bib:ElGamalYHKim2012} for a self-contained treatment).  Since Wagner \etal's original proof of the sum-rate constraint,   other proofs have been proposed based on estimation-theoretic arguments and semidefinite programming (e.g., \cite{wang2010sum}), however all known proofs are  quite complex.  Below, we show that the converse result for the entire rate region is a direct consequence of Theorem \ref{thm:vectorEIPI}, thus unifying the results of \cite{bib:Wagner} and  \cite{bib:Oohama1997} under a common and  succinct inequality.

\begin{theorem}\label{thm:MTSCGaussian}  \cite{bib:Wagner}
Let $\vecX,\vecY = \{X_i,Y_i\}_{i=1}^n$ be independent identically distributed pairs of jointly Gaussian random variables with correlation $\rho$.   Let $\phi_X : \mathbb{R}^n \to \{1,\dots, 2^{nR_X}\}$ and $\phi_Y : \mathbb{R}^n \to \{1,\dots, 2^{nR_Y}\}$, and define 
\begin{align}
d_X &\triangleq \frac{1}{n}\EE\left[ \| \mathbf{X} - \mathbb{E}[\mathbf{X}|\phi_X(\vecX),\phi_Y(\vecY)]  \|^2\right]\\
d_Y &\triangleq \frac{1}{n}\EE\left[ \| \mathbf{Y} - \mathbb{E}[\mathbf{Y}|\phi_X(\vecX),\phi_Y(\vecY)]  \|^2\right].
\end{align}
Then
\begin{align}
R_X &\geq \frac{1}{2}\log \left(\frac{1}{d_X}\left( 1-\rho^2 + \rho^2 2^{-2 R_Y }\right)  \right)\label{oneHelperX}\\
R_Y &\geq \frac{1}{2}\log \left(\frac{1}{d_Y}\left( 1-\rho^2 + \rho^2 2^{-2 R_X }\right)  \right)\label{oneHelperY}\\
R_X+R_Y&\geq \frac{1}{2}\log \frac{(1-\rho^2)\beta(d_X d_Y)}{2d_Xd_Y},\label{recoveredSumRate}
\end{align}
where  and $\beta(D)\triangleq 1 + \sqrt{1 + \frac{4\rho^2 D}{(1-\rho^2)^2}}$.
\end{theorem}

The key ingredient is the following consequence of Theorem \ref{thm:vectorEIPI}:
\begin{proposition}\label{propMTSC}
For $\vecX,\vecY$ as above, 
\begin{align}
2^{-\frac{2}{n}(I(\vecY;U) + I(\vecX;V|U))} \geq \rho^2 \, 2^{-\frac{2}{n}(I(\vecX;U) + I(\vecY;V|U)) } + 1-\rho^2
\end{align}
 for any $U,V$ satisfying $U \to \vecX \to \vecY\to V$.  
\end{proposition}
\begin{proof}
Since mutual information is invariant to scaling, we may assume without loss of generality that $Y_i=\rho X_i+Z_i$, where $X_i\sim N(0,1)$ and $Z_i\sim N(0,1-\rho^2)$, independent of $X_i$.  Now, Theorem \ref{thm:vectorEIPI} implies 
\begin{align}
2^{\frac{2}{n}(h(\vecY|U) - I(\vecX;V|U))} &\geq   2^{\frac{2}{n}(h(\rho \vecX|U) - I(\vecY;V|U))} + 2^{\frac{2}{n}h(\vecZ)}\\
&=\rho^2 2^{\frac{2}{n}(h( \vecX|U) - I(\vecY;V|U))} + (2\pi e)(1-\rho^2).
\end{align}
Since $2^{-\frac{2}{n}h(\vecY)} =   2^{-\frac{2}{n}h(\vecX)} = \frac{1}{2\pi e}$, multiplying through by $\frac{1}{2\pi e}$ establishes the claim.
\end{proof}

\begin{proof}[Proof of Theorem \ref{thm:MTSCGaussian}]  For convenience, put $U = \phi_X(\vecX)$ and $V = \phi_Y(\vecY)$.
Using the Markov relationship $U \to \vecX \to \vecY\to V$, we may  rearrange the exponents in Proposition \ref{propMTSC} to obtain the equivalent inequality 
\begin{align}
2^{-\frac{2}{n} (I(\mathbf{X};U,V) +I(\mathbf{Y};U,V))} \geq 2^{-\frac{2}{n}I(\mathbf{X},\mathbf{Y};U,V) } \left( 1-\rho^2 + \rho^2 2^{-\frac{2}{n}I(\mathbf{X},\mathbf{Y};U,V) }\right) . \label{LRHSmonotone}
\end{align}
The left- and right-hand sides of \eqref{LRHSmonotone} are monotone decreasing in $\frac{1}{n}(I(\mathbf{X};U,V) +I(\mathbf{Y};U,V))$ and $\frac{1}{n}I(\mathbf{X},\mathbf{Y};U,V)$, respectively.  Therefore, if 
\begin{align}
&\frac{1}{n} (I(\mathbf{X};U,V) +I( \mathbf{Y};U,V) ) \geq \frac{1}{2} \log \frac{1}{D} \mbox{~~~and~~~}\frac{1}{n}I(\mathbf{X},\mathbf{Y};U,V) \leq R \label{eqnRD}
\end{align}
for some pair $(R,D)$, then we have $D \geq 2^{-2 R } \left( 1-\rho^2 + \rho^2 2^{-2 R }\right)$,
 which is a quadratic inequality with respect to the term $2^{-2 R }$.  This is easily solved using the quadratic formula to obtain:
\begin{align}
2^{-2R} \leq \frac{2D}{(1-\rho^2)\beta(D)} \quad \Rightarrow \quad R \geq \frac{1}{2}\log \frac{(1-\rho^2)\beta(D)}{2D}, \label{RDineq}
\end{align}
where $\beta(D)\triangleq 1 + \sqrt{1 + \frac{4\rho^2 D}{(1-\rho^2)^2}}$.  Note that Jensen's inequality and the maximum-entropy property of Gaussians imply $ \frac{1}{n}  I(\mathbf{X};U,V)   \geq \frac{1}{2} \log \frac{1}{d_X   }$ and  $ \frac{1}{n}  I(\mathbf{Y};U,V)   \geq \frac{1}{2} \log \frac{1}{d_Y   }$ , so that 
\begin{align}
\frac{1}{n} (I(\mathbf{X};U,V) +I( \mathbf{Y};U,V) ) \geq \frac{1}{2} \log \frac{1}{d_X d_Y },\label{mmseEq}
\end{align}
establishing \eqref{recoveredSumRate} since $\frac{1}{n}I(\mathbf{X},\mathbf{Y};U,V)\leq \frac{1}{n}\left( H(U) + H(V)\right)\leq R_X+R_Y$.  Similarly,    Proposition \ref{propMTSC} implies
\begin{align}
2^{2 R_X + \log d_X } \geq 2^{\frac{2}{n} \left(I(\vecX ;U|V)-I(\vecX ;U,V)\right)}  = 2^{-\frac{2}{n} I(\vecX ;V)}  &\geq (1-\rho^2) + \rho^2 2^{-\frac{2}{n} I( \vecY;V)} \geq (1-\rho^2) + \rho^2 2^{-2 R_Y}.%
\end{align}
Rearranging (and symmetry) yields \eqref{oneHelperX}-\eqref{oneHelperY}.
\end{proof}

\begin{remark}
Proposition \ref{propMTSC} (a special case of Theorem \ref{thm:vectorEIPI}) was first established by the author and Jiao in \cite{courtade2014extremal}.  In fact, Proposition \ref{propMTSC} establishes a stronger result than the converse for the two-terminal Gaussian source coding problem; it shows that the rate regions coincide for the problems when distortion is measured under quadratic loss and logarithmic loss \cite{courtade2014multiterminal, jiao2015justification}.
\end{remark}

\subsection{One-sided Gaussian  Interference Channel}
The one-sided Gaussian interference channel (IC)  (or Z-Gaussian IC) is a discrete memoryless channel, with input-output relationship given by
\begin{align}
Y_1 &= X_1 + W\\
Y_2 &= \alpha Y_1  + X_2+ W_2,
\end{align}
where   $X_i$ and $Y_i$ are the channel inputs and observations corresponding to Encoder $i$ and Decoder $i$, respectively, for $i=1,2$.  Here, $W\sim N(0,1)$ and $W_2\sim N(1-\alpha^2)$ are independent of each other and of  the channel inputs $X_1,X_2$.   We have assumed $|\alpha|<1$ since the setting where $|\alpha|\geq 1$ is referred to as the \emph{strong interference} regime, and the capacity is known to coincide with the Han-Kobayashi inner bound \cite{sato1981capacity, costaInterference, bib:ElGamalYHKim2012, te1981new}.  Observe that we have expressed the one-sided Gaussian IC in \emph{degraded form}, which has  capacity region identical to the corresponding non-degraded version as proved by Costa \cite{costaInterference}.  Despite receiving significant attention from researchers over several decades, the capacity region of the one-sided Gaussian IC remains unknown in the regime of $|\alpha|<1$ described above.  

Having already discussed connections between Costa's EPI  \eqref{eq:costaEPI_cos} and Theorem  \ref{thm:vectorEIPI} above, we remark that  Costa's   EPI  was apparently motivated by  the Gaussian IC \cite{costaInterference}.  Since Theorem \ref{thm:vectorEIPI} generalizes Costa's result, the one-sided Gaussian IC presents itself as a natural application.  Toward this end, we  establish a new multi-letter outer bound to give  a simple demonstration of how Theorem  \ref{thm:vectorEIPI}  might be applied to the  one-sided Gaussian IC. 

\begin{theorem}
$(R_1,R_2) \in \mathscr{C}(\alpha,P_1,P_2)$  only if 
\begin{align}
R_1 &\leq \frac{1}{2}\log(1+P_1)\label{ptpt1}\\
R_2 &\leq \frac{1}{2}\log(1+P_2)\label{ptpt2}\\
2^{-2R_2 + o(1)  }&\geq 2^{- \frac{2}{n} I(X^n_1,X^n_2;Y^n_2)} \sup_{V : Y_1^n \to Y_0^n \to V} \left\{ \alpha^2 2^{2 R_1   -\frac{2}{n} I(Y^n_0;V|Y^n_1)}   +(1-\alpha^2) 2^{\frac{2}{n} I(Y^n_1;V)}\right\},\label{HKsumrateMultLetter}
\end{align}
for some independent $X_1^n,X_2^n$ satisfying the power constraints $\EE[\|X_i^n\|^2]\leq nP_i$, $i=1,2$.
\end{theorem}

\begin{proof}
The only nontrivial inequality to prove is \eqref{HKsumrateMultLetter}.  Thus, we begin by noting that Theorem \ref{thm:vectorEIPI} implies 
\begin{align}
2^{\frac{2}{n} (h(Y_2^n|X_2^n) - I(Y_1^n;V|X_2^n))} &\geq 2^{\frac{2}{n} (h(\alpha Y_1^n|X_2^n) - I(Y_2^n;V|X_2^n))} + 2^{\frac{2}{n}h(W_2^n|X_2^n)}\\
&=\alpha^2 2^{\frac{2}{n} (h(Y_1^n) - I(Y_2^n;V|X_2^n))} + (1-\alpha^2) 2^{\frac{2}{n}h(W^n)}
\end{align}
for all $V$ such that $Y_1^n \to Y_2^n \to V|X_2^n$. 
Since $h(W^n) = h(Y_2^n|X_1^n,X_2^n) = h(Y_1^n|X_1^n)$,  $I(Y_2^n;V |X_2^n) = I(Y_0^n ; V ,X_2^n)$ and $I(Y_1^n;V |X_2^n) = I(Y_1^n ; V ,X_2^n)$, this can be rewritten as
\begin{align}
2^{-\frac{2}{n}  I(X^n_2;Y^n_2)+ \frac{2}{n} I(X^n_1,X^n_2;Y^n_2)} &\geq \sup_{V : Y_1^n \to Y_0^n \to V} \left\{ \alpha^2 2^{\frac{2}{n}  I(X^n_1;Y^n_1) -\frac{2}{n} I(Y^n_0;V|Y^n_1)}   +(1-\alpha^2) 2^{\frac{2}{n} I(Y^n_1;V)}\right\}.%
\end{align}
Therefore,
\begin{align}
2^{-2(R_2 - \epsilon_n )} &\geq 2^{-\frac{2}{n}  I(X^n_2;Y^n_2) } \label{eq:FanoR2}\\
&\geq 2^{- \frac{2}{n} I(X^n_1,X^n_2;Y^n_2)} \sup_{V : Y_1^n \to Y_0^n \to V} \left\{ \alpha^2 2^{\frac{2}{n}  I(X^n_1;Y^n_1) -\frac{2}{n} I(Y^n_0;V|Y^n_1)}   +(1-\alpha^2) 2^{\frac{2}{n} I(Y^n_1;V)}\right\} \label{eq:applyEPIsupV}\\
&\geq 2^{- \frac{2}{n} I(X^n_1,X^n_2;Y^n_2)} \sup_{V : Y_1^n \to Y_0^n \to V} \left\{ \alpha^2 2^{2(R_1-\epsilon_n)  -\frac{2}{n} I(Y^n_0;V|Y^n_1)}   +(1-\alpha^2) 2^{\frac{2}{n} I(Y^n_1;V)}\right\},  \label{eq:FanoR1}
\end{align}
where \eqref{eq:FanoR2} and \eqref{eq:FanoR1} hold for $\epsilon_n \to 0$ due to  Fano's inequality.   Multiplying both sides by $2^{2\epsilon_n}$ proves the claim.
\end{proof}

The Han-Kobayashi achievable region \cite{te1981new, bib:ElGamalYHKim2012} evaluated for Gaussian inputs (without power control) can be expressed as the set of rate pairs $(R_1,R_2)$ satisfying \eqref{ptpt1}, \eqref{ptpt2} and 
\begin{align}
2^{-2 R_2 } &\geq   \frac{ \alpha^2 \,P_2 \,2^{2R_1 }  }{(P_2 + 1-\alpha^2)(1+\alpha^2 P_1 + P_2)}  + \frac{ 1-\alpha^2}{P_2 + 1-\alpha^2}. \label{HKachSum}
\end{align}
Interestingly, \eqref{HKachSum} this takes a similar form to \eqref{HKsumrateMultLetter}; however, it is known that transmission without power control is suboptimal for the Gaussian Z-interference channel in general \cite{costa2011noisebergs, NairCosta2016corner}.  Nevertheless, it may be possible to identify a random variable $V$ in \eqref{HKsumrateMultLetter}, possibly depending on $X_2$, which ultimately improves known bounds.  We leave this for future work.

\subsection{Relationship to Strong Data Processing}\label{subsec:SDPI}
{Strong data processing} inequalities  and their connection to hypercontractivity have garnered much recent  attention \cite{ahlswede1976spreading, nair2014equivalent, anantharam2013hypercontractivity, raginsky2014strong, polyanskiy2014dissipation, bib:calmon, anantharam2013maximal, courtade2013outer}.  For random variables $A,B$, the standard data processing inequality asserts that $I(V;A)\leq I(V;B)$ for any random variable $V$ satisfying $A\to B\to V$.  For $A,B\sim P_{AB}$, it is natural to define the best-possible data processing function 
\begin{align}
g_I(t, P_{AB}) = \sup_{V :A\to B\to V}\left\{ I(V;A) : I(V;B)\leq t \right\},
\end{align}
so that $I(V;A)\leq g_I( I(V;B), P_{AB})\leq 1$ is the sharpest possible data processing inequality for the joint distribution $P_{AB}$.  Thus, Theorem \ref{thm:scalarEIPI} may be rephrased as:
\begin{align}
2^{2(h(Y) - g_I( t, P_{XY}) )} \geq  2^{2(h(X) - t)} + 2^{2h(W)} ~~~~~\forall t\geq0, \label{strongEPI}
\end{align}
where $Y = X+W$.  Given the close relationship between the sharpened EPI \eqref{strongEPI} and {strong data processing}, it might be appropriate to call Theorem \ref{thm:scalarEIPI}  a \emph{strong entropy power inequality}.  In any case, on rearranging, we find the following simple bound on $g_I$ for Gaussian channels:
\begin{corollary}\label{cor:FIbound}
Let $X\sim P_X$ and $Z\sim N(0,1)$ be independent. For $Y  = X+Z$,
\begin{align}
g_I(t,P_{XY}) \leq I(X;Y) - \frac{1}{2}\log\left( 1 +  \frac{1}{2\pi e}2^{2(h(X) - t)} \right). \label{StrongDPIbound}
\end{align}
Moreover, for Gaussian $X$, the inequality \eqref{StrongDPIbound} is an equality.
\end{corollary}

We remark that Calmon, Polyanskiy and Wu \cite{polyanskiy2014dissipation, bib:calmon} have recently considered a complementary setting where they bound the best-possible data processing function defined according to
\begin{align}
F_I(t,\gamma) = \sup \left\{ I(Y;U)  : I(X;U)\leq t, U \to X \to Y \right\},
\end{align}
where $Y=X+Z$, and the supremum is over all $P_{UX}$ such that $\EE[X^2]\leq \gamma$. %

\section{Proof of Main Results}\label{sec:proofs}%
Here we give the main ideas behind proving Theorem  \ref{thm:scalarEIPI}.  Technical details are provided in Section \ref{sec:technicalProofs} and referred to as needed.   For random variables $X,Y\sim P_{XY}$, we write $X|\{Y=y\}$  to denote the random variable $X$ conditional on $\{Y=y\}$.  Note that $X|\{Y=y\}$ is uniquely defined in the sense that different versions of the same are equal $P_Y$-a.e.  A sequence of random variables $X_1, X_2, \dots$ indexed by $n\in \mathbb{N}$ will be denoted by the shorthand $\{X_n\}$, and convergence of $\{X_n\}$ in distribution to a random variable $X_*$ is written $X_n  \xrightarrow{\mathcal{D}} X_*$.

In order to minimize the difference in inequality \eqref{EPItoProve}, we would like to simultaneously minimize the exponent $h(X+W)-I(X;V)$, while maximizing the exponent $h(X)-I(X+W;V)$ over all valid choices of $X,V$.  Toward this end, for a random variable $X\sim P_X$, let $Y$ be defined via the additive Gaussian noise channel $P_{Y|X}$ given by $Y = \sqrt{\snr}X+Z$, where $Z\sim N(0,1)$ and define the family of functionals
\begin{align}
\mathsf{s}_{\lambda}(X,\snr) = -h(X)+\lambda h(Y)+  \inf_{V : X\to Y \to V   } \Big\{ I(Y;V) -\lambda I(X;V)\Big\}
\end{align}
parameterized by $\lambda\geq1$.  Similarly, for $(X,Y,Q) \sim P_{XQ}P_{Y|X}$, define the functional of $P_{XQ}$ 
\begin{align}
\mathsf{s}_{\lambda}(X,\snr|Q) =  -h(X|Q)+\lambda h(Y|Q)+  \inf_{{V : X\to Y \to V |Q } } \Big\{ I(Y;V|Q) -\lambda I(X;V|Q)\Big\}, 
\end{align}
and let $\C\left( \mathsf{s}_{\lambda}(X,\snr)  \right)$ denote the lower convex envelope of $\mathsf{s}_{\lambda}(\cdot,\snr)$ at $X$.  That is, 
\begin{align}
\C\left( \mathsf{s}_{\lambda}(X,\snr)  \right) = \inf_{P_{Q|X}} \mathsf{s}_{\lambda}(X,\snr|Q). 
\end{align}
We consider the optimization problem 
\begin{align}
\mathsf{V}_{\lambda}(\snr) = \inf_{P_X \,:\, \EE[X^2]\leq 1} \C\left( \mathsf{s}_{\lambda}(X,\snr)  \right) = \inf_{P_{XQ} \,:\, \EE[ X^2 ] \leq 1 } \mathsf{s}_{\lambda}(X,\snr|Q). \label{entropOptimization}
\end{align}
\begin{remark}
Note that, in the optimization problem \eqref{entropOptimization}, it suffices to consider $Q\in \mathcal{Q}$ with $|\mathcal{Q}|\leq 2$. Indeed, by Fenchel-Caratheodory-Bunt \cite[Theorem 18$(ii)$]{eggleston1958convexity}, taking $Q$ supported on two points is sufficient to preserve the values of $\EE[X^2]= \sum_{q}p(q) \EE[X^2|Q=q]$ and $\mathsf{s}_{\lambda}(X,\snr|Q) = \sum_{q}p(q) \mathsf{s}_{\lambda}(X,\snr|Q=q)$.
\end{remark}

We have the following explicit characterization of $\mathsf{V}_{\lambda}(\snr)$:
\begin{theorem}\label{thm:explicitVsnr}
\begin{align}
\mathsf{V}_{\lambda}(\snr) =   
\begin{cases}
\frac{1}{2}\left[ \lambda  \log \left( \frac{\lambda 2\pi e }{\lambda-1}\right) - \log \left( \frac{2\pi e }{\lambda-1}\right) + \log (\snr ) \right] & \mbox{if $\snr \geq \frac{1}{\lambda-1}$}\\
 \frac{1}{2}\Big[  \lambda \log\left(2\pi e (1+\snr) \right) -\log\left(2\pi e  \right) \Big] & \mbox{if $\snr \leq \frac{1}{\lambda-1}$}.
\end{cases}
\end{align}
\end{theorem}
The essential idea needed to establish Theorem \ref{thm:explicitVsnr} is that we only need to consider Gaussian random variables in optimization problem \eqref{entropOptimization}.  We establish this  using a weak convergence argument;  the critical ingredients are proved in Sections \ref{subsec:gaussiansequence} and \ref{subsec:weakSemiContinuity}, and respectively assert:
\begin{enumerate} 
\item[\underline{Claim I}:] There exists a sequence $\{X_n,Q_n\}$ satisfying
\begin{align}
\lim_{n\to \infty}  \mathsf{s}_{\lambda}(X_n,\snr|Q_n) &= \mathsf{V}_{\lambda}(\snr) \\
\EE[X_n^2] &\leq 1 ~~\,n= 1,2,\dots 
\end{align}
and $(X_n,Q_n) \xrightarrow{\mathcal{D}} (X_*,Q_*)$,   with $X_*|\{Q_*=q\} \sim N(\mu_q, \sigma_X^2 )$ for  $P_{Q_*}$-a.e. $q$, with  $\sigma_X^2\leq 1$ not depending on $q$. 
\item[\underline{Claim II}:] If $X_n  \xrightarrow{\mathcal{D}} X_* \sim N(\mu,\sigma_X^2) $ and $\sup_n \EE[X_n^2] <\infty$, then
\begin{align}
\liminf_{n\to \infty} \mathsf{s}_{\lambda}(X_n,\snr)  \geq  \mathsf{s}_{\lambda}(X_*,\snr) .
\end{align}
\end{enumerate}
In words, Claim I states that there exists a sequence $\{X_n,Q_n\}$ which approaches the infimum of the optimization problem \eqref{entropOptimization}, with $X_n$ converging weakly to Gaussian.  Claim II notes that the functional $\mathsf{s}_{\lambda}(X,\snr)$ is weakly lower semicontinuous at Gaussian $X$.  Combining the two claims allows us to restrict attention to Gaussian $X$ in optimization problem  \eqref{entropOptimization}.

With these facts in hand, the proof of Theorem \ref{thm:explicitVsnr} follows from elementary calculus and the classical EPI.  We require the following proposition, which is a consequence of the conditional EPI, and a dual formulation of an inequality observed by  Oohama \cite{bib:Oohama1997}.
\begin{proposition}  \label{prop:InfoBottleExplicit}
Let $X\sim N(0,\gamma)$ and $Z\sim N(0,1)$ be independent, and define $Y=  \sqrt{\snr} X+ Z$.  Then for $\lambda \geq 1$,
\begin{align}
\inf_{V : X \to Y\to V} \Big(I(Y;V)-\lambda I(X;V)\Big) = 
\begin{cases}
\frac{1}{2}\left[\log\left((\lambda-1)\gamma \, \snr\right)-\lambda\log 
\left(\frac{\lambda-1}{\lambda}\left(1+ \gamma \, \snr \right)\right)\right] & \mbox{if $\gamma\, \snr \geq \frac{1}{\lambda-1}$}\\
0 & \mbox{if $\gamma\, \snr \leq \frac{1}{\lambda-1}$.}
\end{cases}\notag
\end{align}
\end{proposition}
\begin{proof}
Let $V$ be such that $X \to Y\to V$, and let $X|\{V=v\}, Y|\{V=v\}$ denote the random variables conditioned on $\{V=v\}$.  Since $X,Y$ are jointly Gaussian and $V\to Y \to X$, we have $X|\{V=v\} = \rho Y|\{V=v\} + W$, where $\rho := \frac{\gamma \sqrt{\snr}}{1 + \gamma \,\snr }$ and $W\sim N\left(0,\gamma -  \frac{\gamma^2  \,\snr}{1 + \gamma \,\snr }\right)$ is independent of $Y|\{V=v\}$.  By the entropy power inequality, it holds that
\begin{align}
2^{2h(X|V=v)}\geq  2^{2h(\rho Y|V=v)} + 2^{2 h(W)} =   \frac{\gamma^2\, {\snr}}{(1 + \gamma \,\snr )^2} 2^{2h(Y|V=v)} + 2 \pi e \left( \gamma -  \frac{\gamma^2  \,\snr}{1 + \gamma \,\snr }\right) 
\end{align}
which,  upon applying Jensen's inequality and rearranging,  yields
\begin{align}
2^{-2 I(X;V)}\geq \frac{1 +  \gamma \, \snr  \,2^{-2I(Y;V)}}{ 1 + \gamma \,\snr }. 
\end{align}
It follows that
\begin{align}
I(Y;V)-\lambda I(X;V) &\geq I(Y;V) + \frac{\lambda}{2} \log \left(1 +  \gamma  \, \snr  \, 2^{-2I(Y;V)} \right)  -\frac{\lambda}{2} \log(1 + \gamma \, \snr)\\
&\geq \begin{cases}
\frac{1}{2}\left[\log\left((\lambda-1)\gamma \, \snr\right)-\lambda\log 
\left(\frac{\lambda-1}{\lambda}\left(1+ \gamma \, \snr \right)\right)\right] & \mbox{if $\gamma\, \snr > \frac{1}{\lambda-1}$}\\
0 & \mbox{if $\gamma\, \snr \leq \frac{1}{\lambda-1}$,}
\end{cases}
\end{align}
where the second inequality follows by minimizing over $I(Y;V)\geq 0$.  When $\gamma\, \snr \leq \frac{1}{\lambda-1}$, this is trivially achieved by setting $V = \mathsf{constant}$.  On the other hand, if $\gamma\, \snr > \frac{1}{\lambda-1}$, then it is easy to see that the lower bound is achieved by taking $V = Y + U$, 
where $U\sim N(0,\frac{1 +\gamma \, \snr}{\gamma \, \snr (\lambda -1) -1})$.
\end{proof}

\begin{proof}[Proof of Theorem \ref{thm:explicitVsnr}]
Noting that $\mathsf{s}_{\lambda}(X,\snr)$ is invariant to translations of $\EE[X]$, it follows from Claims I and II that
\begin{align}
\mathsf{V}_{\lambda}(\snr) = \inf_{0 \leq \gamma \leq 1 } \mathsf{s}_{\lambda}(X_{\gamma},\snr),~~~~\mbox{where }X_{\gamma}\sim N(0,\gamma).
\end{align}
Recalling the definition of $\mathsf{s}_{\lambda}(\,\cdot\,,\snr)$, Proposition \ref{prop:InfoBottleExplicit} implies
\begin{align}
 \mathsf{s}_{\lambda}(X_{\gamma},\snr)= \begin{cases}
\frac{1}{2}\left[\lambda\log 
\left(\frac{\lambda 2 \pi e}{\lambda-1} \right) -\log\left(\frac{2\pi e}{\lambda-1}\right)+ \log (\snr ) \right]& \mbox{if $\gamma\, \snr \geq \frac{1}{\lambda-1}$}\\
\frac{1}{2}\left[\lambda\log 
\left(2 \pi e \left(1+ \gamma\, \snr\right)\right)-\log\left(2 \pi e \gamma \right)\right]  & \mbox{if $\gamma\, \snr \leq \frac{1}{\lambda-1}$.}
\end{cases}
\end{align}
Differentiating with respect to the quantity $\gamma $, we find that $\frac{1}{2}\left[\lambda\log 
\left(2 \pi e \left(1+ \gamma\, \snr\right)\right)-\log\left(2 \pi e \gamma \right)\right]$ is decreasing in $\gamma $ provided $\gamma\, \snr \leq \frac{1}{\lambda-1}$.  Therefore, taking $\gamma=1$ minimizes $\mathsf{s}_{\lambda}(X_{\gamma},\snr)$ over the interval $\gamma\in[0,1]$, proving the claim.
\end{proof}

Given the explicit characterization of  $\mathsf{V}_{\lambda}(\snr)$, which is a dual form of inequality \eqref{EPItoProve}, we are now in a position to prove Theorem \ref{thm:scalarEIPI}.
\begin{proof}[Proof of Theorem \ref{thm:scalarEIPI}]
We  first establish \eqref{EPItoProve} under the additional assumption that $\EE[X^2]<\infty$, and generalize at the end via a truncation argument.  Toward this goal, since mutual information is invariant to scaling, it is sufficient to prove that, for $Y = \sqrt{\snr} X+Z$ with $\EE[X^2]\leq 1$ and $Z\sim N(0,1)$ independent of $X$, we have
\begin{align}
2^{2(h(Y) - I(X;V))} \geq \snr \, 2^{2(h(X) - I(Y;V))} + 2^{2h(Z)}\label{preInequality}
\end{align}
for $V$ satisfying $X \to Y \to V$.  Multiplying both sides by $\sigma^2$ and choosing $\snr := \frac{\Var(X)}{\sigma^2}$  gives the desired inequality \eqref{EPItoProve} when $\EE[X^2]<\infty$.   Thus, to prove \eqref{preInequality}, observe by definition of $\mathsf{V}_{\lambda}(\snr)$ that
\begin{align}
-h(X)+  I(Y;V)   \geq \lambda( I(X;V)-h(Y ) ) +\mathsf{V}_{\lambda}(\snr). \label{fnlToMinimize}
\end{align}
Minimizing the RHS over $\lambda$ proves the inequality.  In particular, the RHS of \eqref{fnlToMinimize} is concave in $\lambda$, with derivative given by
\begin{align}
\frac{\partial}{\partial \lambda}  \Big\{ \lambda( I(X;V)-h(Y ) ) +\mathsf{V}_{\lambda}(\snr) \Big\}=
\begin{cases}
I(X;V)-h(Y )  + \frac{1}{2}\log \left(\frac{\lambda 2 \pi e}{\lambda-1} \right)  &\mbox{if $\snr \geq \frac{1}{\lambda-1}$}\\
I(X;V)-h(Y )  + \frac{1}{2}  \log\left(2\pi e (1+\snr) \right) & \mbox{if $\snr < \frac{1}{\lambda-1}$.}
\end{cases}\notag
\end{align}
Since $h(Y )  \leq \frac{1}{2}  \log\left(2\pi e (1+\snr) \right)$ by the maximum entropy property of Gaussians, it follows that $I(X;V)-h(Y )  + \frac{1}{2}  \log\left(2\pi e (1+\snr) \right) \geq 0$, implying that $\frac{\partial}{\partial \lambda}  \Big\{ \lambda( I(X;V)-h(Y ) ) +\mathsf{V}_{\lambda}(\snr) \Big\}=0$ for $\lambda$ satisfying $\snr \geq \frac{1}{\lambda-1}$.  In particular, the  RHS of \eqref{fnlToMinimize} is minimized when $\lambda$ satisfies
\begin{align}
\frac{\lambda}{\lambda-1} = \frac{1}{2\pi e} 2^{-2(I(X;V)-h(Y) )}.
\end{align}
Substituting into \eqref{fnlToMinimize} and recalling that $2^{2 h(Z)} = 2\pi e$ proves \eqref{preInequality}.

Now, we eliminate the assumption that $\EE[X^2]<\infty$. Toward this end, let $X$ have density, let $W$ be Gaussian independent of $X$,  and consider $V$ satisfying $X \to Y \to V$, where $Y=X+W$.  Define  $X_n$  to be the random variable $X$ conditioned on the event  $\{|X| \leq n\}$, let $Y_n = X_n + W$ and define $V_n$ via  $P_{V|Y} : Y_n \mapsto V_n$. Since $X_n$ is bounded, $\EE[X_n^2]<\infty$ so that
\begin{align}
2^{2(h(Y_n) - I(X_n;V_n))} \geq  2^{2(h(X_n) - I(Y_n;V_n))} + 2^{2h(W)}.
\end{align}
It follows by \cite[Lemma 3]{bib:Bobkov} %
that $\lim_{n\to \infty}h(X_n) = h(X)$, provided $h(X)$ exists.  Moreover, since $X_n  \xrightarrow{\mathcal{D}}  X$,  Lemma \ref{lem:weakConvergenceGaussian} (see Section \ref{subsec:GaussianPerturbation}) asserts that  $\lim_{n\to\infty}h(X_n +W) = h(X+W)$, so that $h(Y_n)\to h(Y)$.  It is easy to see that $(X_n,V_n) \xrightarrow{\mathcal{D}} (X,V)$, so  $\liminf_{n\to\infty} I(X_n;V_n)\geq I(X;V)$ by lower semicontinuity of relative entropy.  Finally, the chain rule for mutual information implies 
\begin{align}
I(Y;V) + H( \mathds{1}_{\{|X|\leq n\}} ) \geq I(Y;V | \mathds{1}_{\{|X|\leq n\}} ) \geq I(Y_n;V_n ) \PP\{|X|\leq n\},
\end{align}
giving $\limsup_{n\to\infty} I(Y_n;V_n)\leq I(Y;V)$.  Putting these observations together, we have established
\begin{align}
2^{2(h(Y) - I(X;V))} \geq  2^{2(h(X) - I(Y;V))} + 2^{2h(W)}
\end{align}
as desired.
\end{proof}

\section{Extension to Random Vectors}

The vector generalization of the classical EPI is usually proved by a combination of conditioning, Jensen's inequality and induction (e.g., \cite[Problem 2.9]{bib:ElGamalYHKim2012}). The same argument does not appear to readily apply in generalizing Theorem \ref{thm:scalarEIPI} to its vector version due to complications arising from the Markov constraint $\vecX\to (\vecX+\vecW)\to V$.  However, the desired generalization may be established by noting an additivity property enjoyed by the dual form. 

For a random vector $\vecX \sim P_{\vecX}$, let $\vecY$ be defined via the additive Gaussian noise  channel $\vecY = \Gamma^{1/2} \vecX+\vecZ$, where $\vecZ\sim N(0,I)$ is independent of $\vecX$ and $\Gamma$ is a diagonal matrix with nonnegative diagonal entries.  Analogous to the scalar case,  define the family of functionals
\begin{align}
\mathsf{s}_{\lambda}(\vecX,\Gamma) = -h(\vecX)+\lambda h(\vecY)+  \inf_{V : \vecX\to \vecY \to V   } \Big\{ I(\vecY;V) -\lambda I(\vecX;V)\Big\}
\end{align}
parameterized by $\lambda\geq1$.  Similarly, for $(\vecX,\vecY,Q) \sim P_{\vecX Q}P_{\vecY |\vecX}$, define 
\begin{align}
\mathsf{s}_{\lambda}(\vecX,\Gamma|Q) =  -h(\vecX|Q)+\lambda h(\vecY|Q)+  \inf_{{V : \vecX\to \vecY \to V |Q } } \Big\{ I(\vecY;V|Q) -\lambda I(\vecX;V|Q)\Big\}, 
\end{align}
and  
consider the optimization problem 
\begin{align}
\mathsf{V}_{\lambda}(\Gamma) =   \inf_{P_{\vecX Q} \,:\, \EE[X_i^2]\leq 1, i\in [n]} \mathsf{s}_{\lambda}(\vecX,\Gamma|Q). \label{entropOptimization2}
\end{align}

\begin{theorem} \label{thm:VlamVector}
If $\Gamma = \operatorname{diag}(\snr_1, \snr_2, \dots, \snr_n)$, then
\begin{align}
\mathsf{V}_{\lambda}(\Gamma)   =\sum_{i=1}^n \mathsf{V}_{\lambda}(\snr_i)  .
\end{align}
\end{theorem}
\begin{proof}%
Let $\Gamma$ be a block diagonal matrix with blocks given by $\Gamma = \operatorname{diag}(\Gamma_1,\Gamma_2)$.  Partition $\vecX=(\vecX_1,\vecX_2)$ and $\vecZ=(\vecZ_1,\vecZ_2)$ such that $\vecY_i=\Gamma_i^{1/2}\vecX_i + \vecZ_i$ for $i=1,2$.  Then, for any $V$ such that $\vecX \to \vecY\to V|Q$, it follows from Lemma \ref{lem:additivityofV} (see Section \ref{subsec:gaussiansequence}) that
\begin{align}
\mathsf{s}_{\lambda}(\vecX,\Gamma|Q)  \geq  \mathsf{s}_{\lambda}(\vecX_1 ,\Gamma_1|\vecX_2,Q) + \mathsf{s}_{\lambda}(\vecX_2 ,\Gamma_2|\vecY_1,Q). 
\end{align}
Hence, $\mathsf{V}_{\lambda}(\Gamma) \geq \mathsf{V}_{\lambda}(\Gamma_1) +\mathsf{V}_{\lambda}(\Gamma_2)$ by definition, so induction proves the claim.
\end{proof}

\begin{proof}[Proof of Theorem \ref{thm:vectorEIPI}]
Define $\vecY = \vecX+\vecW$ for convenience.  As in the scalar setting, we establish the unconditional claim (where $Q$ is constant)  under the constraint $\EE [\|\vecX\|^2 ]< \infty$.  The general result follows by a truncation argument exactly as in the scalar setting.  Moreover, we may assume $\Sigma_{\vecW} \succ 0$, else the inequality reduces to $h(\vecY) + I(\vecY;V)\geq h(\vecX) + I(\vecX;V)$, which is trivially true by the data processing inequality and the fact that conditioning reduces entropy.  
 
Thus, due to positive definiteness of $\Sigma_{\vecW}$ and invariance of mutual information under one-one transformations, we may multiply both sides of \eqref{ineqToproveVec} by $|\Sigma_{\vecW}|^{-1/n}$ to obtain the equivalent inequality
\begin{align}
2^{\frac{2}{n}(h(\Sigma_{\vecW}^{-1/2}\vecY) - I(\Sigma_{\vecW}^{-1/2}\vecX;V))} \geq  2^{\frac{2}{n}(h(\Sigma_{\vecW}^{-1/2}\vecX) - I(\Sigma_{\vecW}^{-1/2}\vecY;V))} + 2^{\frac{2}{n}h(\Sigma_{\vecW}^{-1/2}\vecW)}.
\end{align}
However, $\Sigma_{\vecW}^{-1/2}\vecW\sim N(0,I)$ and , $\EE [\| \Sigma_{\vecW}^{-1/2} \vecX\|^2 ]< \infty$ provided $\EE [\|\vecX\|^2 ]< \infty$, so we may assume without loss of generality that $\vecW \sim N(0,I )$ in establishing the unconditional version of \eqref{ineqToproveVec}. 

To simplify further, put $\snr := \max_{1\leq i\leq n}{\EE[X_i^2]}$.  Note that we may assume $\snr >0$, else the claimed inequality is trivial since $h(\vecX)=-\infty$ and $h(\vecY) - I(\vecX;V) \geq h(\vecY) - I(\vecX;\vecY) =h(\vecW)$.  Therefore, \eqref{ineqToproveVec} is  equivalent to 
\begin{align}
2^{\frac{2}{n}(h (\vecY) - I(\vecX ;V))} \geq  \snr\,  2^{\frac{2}{n}(h(\vecX) - I(\vecY;V))} + 2^{\frac{2}{n}h( \vecZ)}
\end{align}
holding for $\vecX\to \vecY\to V$, where $\vecY=\sqrt{\snr} \vecX+\vecZ$,  $\vecZ\sim N(0,I)$ is independent of $\vecX$, and $\max_{1\leq i\leq n}{\EE[X_i^2]}\leq 1$.   This is established exactly as in the proof of Theorem \ref{thm:scalarEIPI}, since $\mathsf{V}_{\lambda}(\snr \cdot I )   =n \mathsf{V}_{\lambda}(\snr)$.
\end{proof}

\section{Proof of Claims I and II}\label{sec:technicalProofs}
This section is dedicated to the proof of Claims I and II of Section \ref{sec:proofs}.  Several of the steps in the proof require properties and characterizations of Gaussian random variables, which are recalled and proved as needed in the first two subsections.  The third subsection is dedicated to the proof of Claim I, and the fourth subsection is dedicated to the proof of Claim II.    

\subsection{Properties of Gaussian Perturbation}\label{subsec:GaussianPerturbation}

We collect below a few facts about random variables that are contaminated by Gaussian noise.  Of particular interest to us will be weakly convergent sequences of random variables, and corresponding continuity properties under perturbation by Gaussian noise.

\begin{lemma}\label{lem:GaussSmooth} \cite[Lemma 5.1.3]{bryc2012normal}
If $X,Z$ are independent random variables and $Z$ is normal, then $X+Z$ has a non-vanishing probability density function which has derivatives of all orders.
\end{lemma}

\begin{lemma}\label{lem:weakConvergenceGaussian} \cite[Propositions 16 and 18]{GengNair}
Let $\vecX_n  \xrightarrow{\mathcal{D}} \vecX_*$ with $\sup_n \EE[\|\vecX_n\|^2] <\infty$, and let $\vecZ\sim N(0,\sigma^2I )$ be a non-degenerate Gaussian, independent of $\{\vecX_n\},\vecX_*$.  Let $\vecY_n = \vecX_n+\vecZ$ and $\vecY_* = \vecX_*+\vecZ$.  Finally, let $f_n(\mathbf{y})$ and $f_*(\mathbf{y})$ denote the density of $\vecY_n$ and $\vecY_*$, respectively.  Then
\begin{enumerate}[~~1.]
\item $\vecY_n  \xrightarrow{\mathcal{D}} \vecY_*$
\item $\|f_n(\mathbf{y}) - f_*(\mathbf{y}) \|_{\infty}\to 0$
\item $h(\vecY_n) \to h(\vecY_*)$.
\end{enumerate}
\end{lemma}

\begin{lemma}\label{lem:WLSC_condMI}
Suppose $(X_{1,n},X_{2,n}) \xrightarrow{\mathcal{D}} (X_{1,*},X_{2,*})$ with $\sup_n \EE[X_{i,n}^2]<\infty$ for $i=1,2$. Let $(Z_1,Z_2)\sim N(0,\sigma^2 I)$ be pairwise independent of $(X_{1,n},X_{2,n})$ and  $(X_{1,*},X_{2,*})$, and, for $i=1,2$, define $Y_{i,n} = X_{i,n} + Z_i$ and $Y_{i,*} = X_{i,*} + Z_i$. Then $(Y_{1,n},Y_{2,n}) \xrightarrow{\mathcal{D}} (Y_{1,*},Y_{2,*})$ and
\begin{align}
\liminf_{n\to\infty} I(X_{1,n};X_{2,n}|Y_{1,n},Y_{2,n}) \geq I(X_{1,*};X_{2,*}|Y_{1,*},Y_{2,*}).  \label{wLSC}
\end{align}
\end{lemma}
\begin{proof}
The fact that $(Y_{1,n},Y_{2,n}) \xrightarrow{\mathcal{D}} (Y_{1,*},Y_{2,*})$ follows from Lemma \ref{lem:weakConvergenceGaussian}.  Lemma \ref{lem:weakConvergenceGaussian} also establishes that 
\begin{align}
h(Y_{1,n},Y_{2,n}) \to h(Y_{1,*},Y_{2,*} ).\label{limH}
\end{align}
On account of the Markov chains $(X_{2,n},Y_{2,n})\to X_{1,n} \to Y_{1,n}$ and $(X_{1,n},Y_{1,n})\to X_{2,n} \to Y_{2,n}$, we have the  identity
\begin{align}
I(X_{1,n};X_{2,n}|Y_{1,n},Y_{2,n}) &= I(X_{1,n},Y_{2,n};Y_{1,n},X_{2,n}) - I(X_{1,n},X_{2,n};Y_{1,n},Y_{2,n}).\label{MIident}
\end{align}
Observe that $\liminf_{n\to \infty}I(X_{1,n},Y_{2,n};Y_{1,n},X_{2,n}) \geq  I(X_{1,*},Y_{2,*};Y_{1,*},X_{2,*})$ due to lower semicontinuity of relative entropy, and $\lim_{n\to \infty}  I(X_{1,n},X_{2,n};Y_{1,n},Y_{2,n})= I(X_{1,*},X_{2,*};Y_{1,*},Y_{2,*})$ due to \eqref{limH} and the fact that  $h(Y_{1,*},Y_{2,*}|X_{1,*},X_{2,*}) = h(Y_{1,n},Y_{2,n}|X_{1,n},X_{2,n})=h(Z_1,Z_2)$ is constant.  Thus, \eqref{wLSC} is proved by applying the identity \eqref{MIident} again for $(X_{1,*},X_{2,*},Y_{1,*},Y_{2,*})$.
\end{proof}

\begin{lemma}\label{lem:MIbounds} Let $ \{Y_n\},  Y_*$ be as in Lemma \ref{lem:weakConvergenceGaussian}.   Fix $b>0$ and a channel $P_{V|Y}$, and define $\{V_n\},V_*$ according to $P_{V|Y}: Y_n \mapsto V_n$  and $P_{V|Y}: Y_* \mapsto V_*$.  There exists a sequence $\{\epsilon_n\}$ depending on $b$ and $\{Y_n\}$, but not on $P_{V|Y}$, satisfying $\lim_{n\to\infty}\epsilon_n=0$ and
\begin{align}
I(V_n; Y_n\,| ~|Y_n| \leq b) &\leq (1+\epsilon_n)^2  I(V_*;Y_*\,| ~|Y_*|\leq b)  -  (1+\epsilon_n)^2 \log (1-\epsilon_n )^2,\mbox{}\label{MIineq1}\\
I(V_n; Y_n\,| ~|Y_n| \leq b) &\geq (1-\epsilon_n)^2  I(V_*;Y_*\,|~|Y_*| \leq b)  -  (1-\epsilon_n)^2 \log (1+\epsilon_n )^2,\mbox{}\label{MIineq2}\\
\left|\frac{\PP(|Y_n| \leq b)}{\PP(|Y_*| \leq  b)} - 1\right|&\leq \epsilon_n.\label{PB_ineq}
\end{align}
\end{lemma}

\begin{proof} Let $f_n(y)$ and $f_*(y)$ denote the density of $Y_n$ and $Y_*$, respectively.  By Lemma \ref{lem:GaussSmooth}, the density $f_*$ is continuous and does not vanish,  and is therefore bounded away from zero on the interval $B=[-b,b]$.  By Lemma \ref{lem:weakConvergenceGaussian}, $\|f_n(y) - f_*(y) \|_{\infty}\to 0$, so it follows that 
\begin{align}
\sup_{y\in {B} }\left| 1 - \frac{f_n(y)}{f_*(y)}\right| \leq \epsilon_n  \mbox{~~~~and~~~~} \sup_{y\in {B} }\left| 1 - \frac{f_*(y)}{f_n(y)}\right| \leq \epsilon_n
\end{align} 
for some $\epsilon_n \to 0$ as $n\to \infty$ (note that $\epsilon_n$ does not depend on $P_{V|Y}$).  As a consequence, 
\begin{align}
\PP(Y_* \in {B}) = \int_B f_*(y) dy \leq (1+\epsilon_n)\int_B f_n(y) dy  = (1+\epsilon_n) \PP(Y_n \in {B}).\label{eqAlongLines}
\end{align}
Hence, for $y\in B$, the conditional densities of $Y_n |\{Y_n \in {B}\}$ and $Y_* |\{Y_* \in {B}\}$ satisfy
\begin{align}
\frac{f_{Y_n |\{Y_n \in {B}\}}(y)}{f_{Y_* |\{Y_* \in {B}\}}(y)} = \frac{f_n(y)}{f_*(y)} \cdot \frac{\PP(Y_* \in {B})}{\PP(Y_n \in {B})}\leq (1+\epsilon_n)^2.
\end{align}
By a symmetric argument, ${f_{Y_n |\{Y_n \in {B}\}}(y)} \geq (1-\epsilon_n)^2f_{Y_* |\{Y_* \in {B}\}}(y)$ for all $y\in B$.
Therefore, for any Borel set $A$, 
\begin{align}
\PP( V_n \in A | Y_n \in B ) = \int_B \int_A P_{V|Y=y}(dv) f_{Y_n |\{Y_n \in {B}\}}(y) dy &\geq (1-\epsilon_n)^2 \int_B \int_A P_{V|Y=y}(dv) f_{Y_* |\{Y_* \in {B}\}}(y) dy\\
&= (1-\epsilon_n)^2 \, \PP( V_* \in A | Y_* \in B ).
\end{align}
As a consequence, $\frac{dP_{V_n | Y_n \in B}}{dP_{V_* | Y_* \in B}}(v)\geq  (1-\epsilon_n)^2 $.  Combining the above observations we have
\begin{align}
I(V_n;Y_n|Y_n \in B) &= \int_B \int  f_{Y_n |\{Y_n \in {B}\}}(y) 
 \log \left( \frac{ dP_{V|Y=y} }{  dP_{V_n | Y_n \in B}  }(v) \right)  P_{V|Y=y}( dv  )  dy\\
&\leq (1+\epsilon_n)^2 \int_B \int f_{Y_* |\{Y_* \in {B}\}}(y) 
\log \left( \frac{1}{(1-\epsilon_n )^2 }\frac{ dP_{V|Y=y}  }{ dP_{V_* | Y_* \in B}  }(v) \right) P_{V|Y=y}(dv) dy\\
&= (1+\epsilon_n)^2  I(V_*;Y_*|Y_* \in B)  -  (1+\epsilon_n)^2 \log (1-\epsilon_n )^2.
\end{align}
By a symmetric argument, we also have
\begin{align}
I(V_n; Y_n|Y_n \in B) &\geq (1-\epsilon_n)^2  I(V_*;Y_*|Y_* \in B)  -  (1-\epsilon_n)^2 \log (1+\epsilon_n )^2,
\end{align}
which proves \eqref{MIineq1}-\eqref{MIineq2}.  Inequality \eqref{PB_ineq} is established by the same logic as \eqref{eqAlongLines}.
\end{proof}

\begin{lemma}\label{lem:MItailBound}
Let $X\sim P_X$ and let $Z\sim N(0,\sigma^2)$ be a non-degenerate Gaussian, independent of $X$. It holds that
\begin{align}
\lim_{b\to \infty} \PP(|X|>b) I(X;X+Z~| ~|X|>b) =0.
\end{align}
\end{lemma}
\begin{proof}
The proof follows that of \cite[Theorem 6]{WuVerdu}.  
By lower semicontinuity of relative entropy, we have
\begin{align}
\liminf_{b\to\infty }I(X;X+Z~| ~|X| \leq b) \geq I(X;X+Z).
\end{align}
Also,
\begin{align}
I(X;X+Z)\geq \PP(|X|\leq b) I(X;X+Z~| ~|X|\leq b),
\end{align}
so that $\lim_{b\to\infty }I(X;X+Z~| ~|X| \leq b) = \lim_{b\to\infty } \PP(|X|\leq b) I(X;X+Z~| ~|X|\leq b) =  I(X;X+Z).$  By the chain rule
\begin{align}
\PP(|X|>b) I(X;X+Z~| ~|X|>b) = I(X;X+Z) - I(\mathds{1}_{\{|X|\leq b\}} ; X+Z) - \PP(|X|\leq b) I(X;X+Z~| ~|X|\leq b) ,\notag
\end{align}
so  the claim is proved since $I(\mathds{1}_{\{|X|\leq b\}} ; X+Z)$ vanishes as $b\to \infty$.
\end{proof}

\subsection{Characterizations of Gaussian Random Variables}
The goal of this subsection is to establish the following characterization of Gaussian random variables:
\begin{lemma}\label{lem:GaussianCharMI}
Suppose $(X_{1,n},X_{2,n}) \xrightarrow{\mathcal{D}} (X_{1,*},X_{2,*})$ with $\sup_n \EE[X_{i,n}^2]<\infty$ for $i=1,2$. Let $(Z_1,Z_2)\sim N(0,\sigma^2 I)$ be pairwise independent of $(X_{1,n},X_{2,n})$ and, for $i=1,2$, define $Y_{i,n} = X_{i,n} + Z_i$. If $X_{1,n},X_{2,n}$ are independent and 
\begin{align}
\liminf_{n\to\infty} I(X_{1,n} + X_{2,n} ; X_{1,n} - X_{2,n}|Y_{1,n},Y_{2,n})=0,\label{vanishAssume}
\end{align}
then $X_{1,*},X_{2,*}$ are independent Gaussian random variables with identical variances.
\end{lemma}

We require two facts.  First, a fundamental result of Bernstein \cite{bernstein1941property} asserts the following:
\begin{lemma}\label{thm:Bernstein} \cite[Theorem 5.1.1]{bryc2012normal}
If $X_1,X_2$ are independent random variables such that $X_1 + X_2$ and $X_1-X_2$ are independent, then $X_1$ and $X_2$ are normal, with identical variances.
\end{lemma}
\begin{remark}
Formally, Bernstein's theorem does not comment on the identical variances of $X_1,X_2$.  However, assuming without loss of generality that $X_1,X_2$ are zero-mean, the observation that  $X_1$ and $X_2$ have identical variances is immediate  since $\EE[X_1^2]-\EE[X_2^2] =\EE[(X_1-X_2)(X_1+X_2)] = 0$.  This fact was  explicitly noted by  Geng and Nair \cite{GengNair}.
\end{remark}

Second, we will need the following   observation:
\begin{lemma}\label{channelBern}
Let $Y=X+Z$, where $Z\sim N(0,\sigma^2)$ is a non-degenerate Gaussian, independent of $X$.  If $X|\{Y=y\}$ is normal for $P_Y$-a.e. $y$, with variance $\sigma_X^2$ not depending on $y$, then $X$ is normal with variance $\frac{\sigma^2 \sigma_X^2}{\sigma^2-\sigma_X^2}$.
\end{lemma}
\begin{proof}
If $X|\{Y=y\}$ is normal for $P_Y$-a.e. $y$ with variance $\sigma_X^2$ not depending on $Y$, then $X = \EE[X|Y] +W$ a.s., where $W\sim N(0,\sigma_X^2)$ is independent of $Y$.  In particular, $X$ has density $f_X$ by Lemma \ref{lem:GaussSmooth}.  Also, by Lemma \ref{lem:GaussSmooth}, $Y$ has density $f_Y$.  The conditional density $f_{Y|X}$ exists and is Gaussian by definition, and $f_{X|Y}$ is a valid Gaussian density  for $P_Y$-a.e. $y$, with corresponding variance $\sigma_X^2$ not depending on $y$.  Thus, we have
\begin{align}
\log f_X(x) = \log f_Y(y) + \log f_{Y|X}(y|x) - \log f_{X|Y}(x|y).\label{densities}
\end{align}
The key observation is that the RHS of \eqref{densities} is a quadratic function in $x$.  Since $f_X$ is a density and must integrate to unity, it must therefore be Gaussian.  Direct computation reveals that $X$ has variance $\frac{\sigma^2 \sigma_X^2}{\sigma^2-\sigma_X^2}$. 
\end{proof}

\begin{proof}[Proof of Lemma \ref{lem:GaussianCharMI}]
 Let $Y_{i,*}$ be as in the statement of Lemma \ref{lem:WLSC_condMI}, and recall that  the same lemma asserts $(Y_{1,n},Y_{2,n}) \xrightarrow{\mathcal{D}} (Y_{1,*},Y_{2,*})$.  By definition of $Z_1,Z_2$, the random variables $(Z_1+Z_2)$ and $(Z_1-Z_2)$ are independent and Gaussian with respective variances $2\sigma^2$.  Thus, noting that assumption \eqref{vanishAssume} is equivalent to 
 \begin{align}
\liminf_{n\to\infty} I(X_{1,n} + X_{2,n} ; X_{1,n} - X_{2,n}|Y_{1,n}+Y_{2,n}, Y_{1,n}-Y_{2,n})=0,
\end{align}
we may apply Lemma \ref{lem:WLSC_condMI} to the sequences  $\{X_{1,n} + X_{2,n} , X_{1,n} - X_{2,n}\}$ and $\{Y_{1,n} + Y_{2,n} , Y_{1,n} - Y_{2,n}\}$ to obtain 
\begin{align}
I(X_{1,*}+X_{2,*};X_{1,*}-X_{2,*}|Y_{1,*},Y_{2,*})=I(X_{1,*}+X_{2,*};X_{1,*}-X_{2,*}|Y_{1,*}+Y_{2,*}, Y_{1,*}-Y_{2,*})=0.
\end{align}
Using independence of $X_{1,n},X_{2,n}$, Lemma \ref{lem:WLSC_condMI} applied directly yields
 \begin{align}
I(X_{1,*} ; X_{2,*}|Y_{1,*},Y_{2,*})=0.
\end{align}
In particular, for $P_{Y_{1,*} Y_{2,*}}$-a.e. $y_1,y_2$, the random variables  $X_{1,*}|\{Y_{1,*},Y_{2,*}=y_1,y_2\}$ and $X_{2,*}|\{Y_{1,*},Y_{2,*}=y_1,y_2\}$ are independent, and $(X_{1,*}+X_{2,*})|\{Y_{1,*},Y_{2,*}=y_1,y_2\}$ and $(X_{1,*}-X_{2,*})|\{Y_{1,*},Y_{2,*}=y_1,y_2\}$ are independent.  Therefore, Lemma \ref{thm:Bernstein} implies that $X_{1,*}|\{Y_{1,*},Y_{2,*}=y_1,y_2\}$ and $X_{2,*}|\{Y_{1,*},Y_{2,*}=y_1,y_2\}$ are normal with identical variances.  Starting with the third claim of Lemma \ref{lem:weakConvergenceGaussian} and applying lower semicontinuity of relative entropy, we observe
\begin{align}
I(X_{1,*};Y_{1,*}) = \lim_{n\to \infty} I(X_{1,n};Y_{1,n}) 
&= \lim_{n\to \infty}I(X_{1,n};Y_{1,n},Y_{2,n}) \\
&\geq I(X_{1,*};Y_{1,*},Y_{2,*}) \\
&=I(X_{1,*};Y_{1,*} ) +I(X_{1,*};Y_{2,*}|Y_{1,*}),
\end{align}
so it follows  that $X_{1,*} \to Y_{1,*} \to Y_{2,*}$, and therefore $X_{1,*}|\{Y_{1,*},Y_{2,*}=y_1,y_2\} \sim X_{1,*}|\{Y_{1,*} =y_1 \}$.  Similarly, $X_{2,*}|\{Y_{1,*},Y_{2,*}=y_1,y_2\} \sim X_{2,*}|\{Y_{2,*} =y_2 \}$.  So, we may conclude that the random variables $X_{1,*}|\{Y_{1,*} =y_1 \}$ and $X_{2,*}|\{Y_{2,*} =y_2 \}$ are normal, with identical variances not depending on $y_1,y_2$.  Invoking Lemma \ref{channelBern}, we find  that both $X_{1,*}$ and $X_{2,*}$ are normal with identical variances, completing the proof.

\end{proof}

\subsection{Existence of sequences satisfying $\lim_{n\to \infty}  \mathsf{s}_{\lambda}(X_n,\snr|Q_n)  = \mathsf{V}_{\lambda}(\snr)$ that converge weakly to Gaussian}\label{subsec:gaussiansequence}

The goal of this section is to prove the following result, which was the first essential ingredient needed for the proof of Theorem \ref{thm:explicitVsnr} (i.e., Claim I).  
\begin{lemma}\label{lem:ExistsGaussianSequence}
There exists a sequence $\{X_n,Q_n\}$ satisfying
\begin{align}
\lim_{n\to \infty}  \mathsf{s}_{\lambda}(X_n,\snr|Q_n) &= \mathsf{V}_{\lambda}(\snr)  \label{approachOpt}\\
\EE[X_n^2] &\leq 1 ~~\,n\geq 1 \label{Qalph}
\end{align}
and $(X_n,Q_n) \xrightarrow{\mathcal{D}} (X_*,Q_*)$,   with $X_*|\{Q_*=q\} \sim N(\mu_q, \sigma_X^2 )$ for  $P_{Q_*}$-a.e. $q$, with  $\sigma_X^2\leq 1$ not depending on $q$. 
\end{lemma}

A rough outline of the proof is as follows:  We first establish a superadditivity property of $\mathsf{s}_{\lambda}(X,\snr|Q)$, and then exploit this property in conjunction with the characterization of Gaussians proved in Lemma \ref{lem:GaussianCharMI} to verify the existence of sequence  $\{X_n,Q_n\}$ satisfying $\lim_{n\to \infty}  \mathsf{s}_{\lambda}(X_n,\snr|Q_n) = \mathsf{V}_{\lambda}(\snr)$ which converges weakly to Gaussian.
  We begin with a straightforward observation:
\begin{lemma}\label{lem:additivityofV}  Let $\vecX=(X_1,X_2)$, $\vecY=(Y_1,Y_2)$, and $Q$ have joint distribution $P_{\vecX \vecY Q} = P_{X_1X_2 Q} P_{Y_1|X_1}P_{Y_2|X_2}$.  If $V$ satisfies $\vecX \to \vecY \to V|Q$, then for $\lambda\geq 1$, we have
\begin{align}
&I(Y_1 ,Y_2 ;V  |Q) - h(X_1,X_2|Q) - \lambda\left( I(X_1,X_2;V|Q)-h(Y_1,Y_2|Q) \right) \notag\\
&\geq~~ I(Y_1 ;V | X_2,Q) - h(X_1|X_2,Q) - \lambda\left(I(X_1;V|X_2,Q) - h(Y_1|X_2,Q)\right) \label{additiveIneq}\\
&~~+  I(Y_2 ;V |Y_1,Q) - h(X_2|Y_1,Q) -\lambda\left(  I(X_2;V|Y_1,Q)- h(Y_2|Y_1,Q)  \right). \notag
\end{align}
Moreover,  $X_1 \to Y_1\to V|(X_2,Q)$ and $X_2 \to Y_2\to V|(Y_1,Q)$.
\end{lemma}
\begin{proof}
The second claim is straightforward.  Indeed, using $P_{\vecX \vecY Q} = P_{X_1X_2 Q} P_{Y_1|X_1}P_{Y_2|X_2}$, we can factor the joint distribution of $(\vecX,\vecY, V,Q)$ as $P_{\vecX \vecY V Q} = P_{X_1X_2 Q} P_{Y_1|X_1} P_{Y_2|X_2}P_{V|Y_1 Y_2 Q} = P_{X_1X_2 Q} P_{Y_1|X_1}P_{Y_2V|Y_1 X_2 Q}$.  Marginalizing over $Y_2$, we find that $X_1 \to Y_1\to V|(X_2,Q)$.  The symmetric Markov chain follows similarly by writing $P_{\vecX \vecY V, Q} = P_{X_1X_2 Y_1 Q} P_{Y_2|X_2} P_{V|Y_1 Y_2 Q}$ and marginalizing over $X_1$. 

To prove the claimed inequality, note the following identities:
\begin{align}
 & I(Y_1 ,Y_2 ;V  |Q) - h(X_1,X_2|Q) \notag\\
  &=   I(Y_1 ;V  |Q)  +   I(Y_2 ;V |Q,Y_1) - h(X_2|Q) - h(X_1|Q,X_2)\\
  &=   I(Y_1 ;V  |Q)  +   I(Y_2 ;V |Q,Y_1) - h(X_2|Q,Y_1) - h(X_1|Q,X_2) -I(X_2;Y_1|Q)\\
  &=  I(Y_1 ;V | Q,X_2)  +   I(Y_2 ;V |Q,Y_1) - h(X_2|Q,Y_1) - h(X_1|Q,X_2) -I(X_2;Y_1|Q,V),
\end{align}
and
\begin{align}
&I(X_1,X_2;V|Q)-h(Y_1,Y_2|Q) \notag\\
&= I(X_2;V|Q)+I(X_1;V|Q,X_2) - h(Y_1|Q) - h(Y_2|Q,Y_1)\\
&=I(X_2;V|Q)+I(X_1;V|Q,X_2) - h(Y_1|Q,X_2) - h(Y_2|Q,Y_1) -I(X_2;Y_1|Q)\\
&=I(X_2;V|Q,Y_1)+I(X_1;V|Q,X_2) - h(Y_1|Q,X_2) - h(Y_2|Q,Y_1) -I(X_2;Y_1|Q,V).
\end{align}
Therefore, 
\begin{align}
&I(Y_1 ,Y_2 ;V  |Q) - h(X_1,X_2|Q) - \lambda\left( I(X_1,X_2;V|Q)-h(Y_1,Y_2|Q) \right) \\
&=~~ I(Y_1 ;V | X_2,Q) - h(X_1|X_2,Q) - \lambda\left(I(X_1;V|X_2,Q) - h(Y_1|X_2,Q)\right) \\
&~~+  I(Y_2 ;V |Y_1,Q) - h(X_2|Y_1,Q) -\lambda\left(  I(X_2;V|Y_1,Q)- h(Y_2|Y_1,Q)  \right)  \notag \\
&~~+(\lambda-1)I(X_2;Y_1|V,Q),\notag
\end{align}
which proves the inequality \eqref{additiveIneq} since $\lambda\geq1$.
\end{proof}

Lemma \ref{lem:additivityofV} leads to the desired superadditivity property of $\mathsf{s}_{\lambda}(X,\snr|Q)$:
\begin{lemma}\label{lem:Gaussiansubadditivity}
Let  $P_{Y|X}$ be the Gaussian channel $Y = \sqrt{\snr}X + Z$, where $Z\sim N(0,1)$ is independent of $X$.  Now, suppose $(X,Y,Q)\sim P_{XQ}P_{Y|X}$, and let $(X_1,Y_1,Q_1)$ and $(X_2,Y_2,Q_2)$ denote two independent copies of $(X,Y,Q)$. Define 
\begin{align}
&\xh{+} = \frac{X_{1} + X_{2}}{\sqrt{2}} &\xh{-} = \frac{X_{1} - X_{2}}{\sqrt{2}},\label{hatTransform}
\end{align}
and in a similar manner, define $\yh{+}, \yh{-}$. Letting $\mathbf{Q}=(Q_{1},Q_{2})$, we have for $\lambda\geq 1$
\begin{align}
2 \mathsf{s}_{\lambda}(X,\snr|Q) \geq \mathsf{s}_{\lambda}(X_{+},\snr| X_{-}, \mathbf{Q}) + \mathsf{s}_{\lambda}(X_{-},\snr| Y_{+} ,\mathbf{Q}) \label{firstSineq}
\end{align}
and 
\begin{align}
2 \mathsf{s}_{\lambda}(X,\snr|Q) \geq \mathsf{s}_{\lambda}(X_{+},\snr| Y_{-} ,\mathbf{Q}) + \mathsf{s}_{\lambda}(X_{-},\snr| X_{+} ,\mathbf{Q}).
\end{align}
\end{lemma}
\begin{proof}
The crucial observation is that the unitary transformation $(Y_1,Y_2) \mapsto (Y_+,Y_-)$ preserves the Gaussian nature of the channel.  That is, if $Y_i =  \sqrt{\snr}X_i + Z_i$, then   $Y_+ = \sqrt{\snr}X_+ + \frac{1}{\sqrt{2}}(Z_1+Z_2)$ and $Y_- = \sqrt{\snr}X_- + \frac{1}{\sqrt{2}}(Z_1-Z_2)$, where the pair $(\frac{1}{\sqrt{2}}(Z_1+Z_2), \frac{1}{\sqrt{2}}(Z_1-Z_2))$ is equal in distribution to $(Z_1,Z_2)$.  

Thus, consider an arbitrary $V$ satisfying $(X_+,X_-) \to (Y_+,Y_-) \to V|\mathbf{Q}$.  By Lemma \ref{lem:additivityofV} and the above observation, we have
\begin{align}
&I(Y_1 ,Y_2 ;V  |\mathbf{Q}) - h(X_1,X_2|\mathbf{Q}) - \lambda\left( I(X_1,X_2;V|\mathbf{Q})-h(Y_1,Y_2|\mathbf{Q}) \right) \notag\\
&= ~~I(Y_+ ,Y_- ;V  |\mathbf{Q}) - h(X_+,X_-|\mathbf{Q}) - \lambda\left( I(X_+,X_-;V|\mathbf{Q})-h(Y_+,Y_-|\mathbf{Q}) \right) \\
&\geq~~ I(Y_+ ;V | X_-,Q) - h(X_+|X_-,\mathbf{Q}) - \lambda\left(I(X_+;V|X_-,\mathbf{Q}) - h(Y_+|X_-,\mathbf{Q})\right) \\
&~~+  I(Y_- ;V |Y_+,\mathbf{Q}) - h(X_-|Y_+,\mathbf{Q}) -\lambda\left(  I(X_-;V|Y_+,\mathbf{Q})- h(Y_-|Y_+,\mathbf{Q})  \right)\notag\\
&\geq \mathsf{s}_{\lambda}(X_{+},\snr| X_{-} ,\mathbf{Q}) + \mathsf{s}_{\lambda}(X_{-},\snr| Y_{+}, \mathbf{Q}).
\end{align}
This proves \eqref{firstSineq} since
\begin{align}
&\inf_{V : \vecX \to (\vecY,\mathbf{Q}) \to V } I(Y_1 ,Y_2 ;V  |\mathbf{Q}) - h(X_1,X_2|\mathbf{Q}) - \lambda\left( I(X_1,X_2;V|\mathbf{Q})-h(Y_1,Y_2|\mathbf{Q}) \right)\\
&\leq \sum_{i=1}^2 \inf_{V : X_i \to Y_i \to V|Q_i } I(Y_i  ;V  |Q_i) - h(X_i|Q_i) - \lambda\left( I(X_i;V|Q_i)-h(Y_i|Q_i) \right)\\
&=2\, \mathsf{s}_{\lambda}(X,\snr|Q) ,
\end{align}
where the inequality follows since the infimum is taken over a smaller set. 
\end{proof}

\begin{remark}
In some sense, Lemma \ref{lem:Gaussiansubadditivity} is the key to the whole proof.  The subadditivity property ultimately implies that the optimizing distribution  in  optimization problem \eqref{entropOptimization} is rotationally invariant, and therefore Gaussian. This idea was introduced to the information theory literature by Geng and Nair \cite{GengNair}, but has origins in a `doubling trick' which has been used to great success in the literature on functional inequalities \cite{lieb1990gaussian, carlen1991superadditivity} and has been attributed to K.~Ball \cite{barthe1998optimal}.  The reader is referred to \cite{carlen2009subadditivity, lccv2015} for a detailed discussion of the duality between extremisation of information measures and functional inequalities.
\end{remark}

We are now ready to prove Lemma \ref{lem:ExistsGaussianSequence}.

\begin{proof}[Proof of Lemma \ref{lem:ExistsGaussianSequence}]
For convenience, we will refer to any sequence  $\{X_n, Q_n\}$ satisfying \eqref{approachOpt}-\eqref{Qalph} as \emph{admissible}. Since $\mathsf{s}_{\lambda}(X_n,\snr|Q_n)$ is invariant to translations of the mean of $X_n$, we may restrict our attention to admissible sequences satisfying $\EE[X_n]=0$ without any loss of generality.

Begin by letting $\{X_n, Q_n\}$ be an admissible sequence  with the property that
 \begin{align}
 \lim_{n\to \infty}  \left( h(Y_n|Q_n) - h(X_n|Q_n)\right) &\leq  \liminf_{n\to \infty} \left( h({Y}'_n|{Q}'_n) - h({X}'_n|{Q}'_n)\right) \label{smallestH}
 \end{align}
for any other admissible sequence $\{{X}'_n, {Q}'_n\}$.  Clearly, such a sequence can always be constructed by a diagonalization argument.  Moreover, the LHS of \eqref{smallestH} must be finite.  To see this, note first that $h(Y_n|Q_n) - h(X_n|Q_n)\geq 0$ since conditioning reduces entropy.  On the other hand, $\mathsf{s}_{\lambda}(X_n,\snr|Q_n) < \mathsf{V}_{\lambda}(\snr) +1$ for $n$ sufficiently large.  Hence, there is some $V_n$ satisfying $X_n\to Y_n \to V_n |Q_n$ for which
\begin{align}
h(Y_n|Q_n) - h(X_n|Q_n) &\leq \mathsf{V}_{\lambda}(\snr) +1 + \lambda I(X_n;V_n|Q_n) - I(Y_n;V_n|Q_n) - (\lambda-1) h(Y_n|Q_n)\\
&\leq \mathsf{V}_{\lambda}(\snr) +1 + (\lambda-1) I(X_n;V_n|Q_n) -  (\lambda-1) h(Y_n|Q_n)\label{dpi1}\\
&\leq \mathsf{V}_{\lambda}(\snr) +1 + (\lambda-1) I(X_n;Y_n|Q_n) - (\lambda-1) h(Y_n|Q_n)\label{dpi2}\\
&= \mathsf{V}_{\lambda}(\snr) +1 - (\lambda-1) h(Y_n|X_n) ,
\end{align} 
where \eqref{dpi1} and \eqref{dpi2} are both due to the data processing inequality.  Since $ \mathsf{V}_{\lambda}(\snr)<\infty$ trivially and $h(Y_n|X_n) = h(Z)=\frac{1}{2}\log 2\pi e$,  we conclude that the LHS of \eqref{smallestH} is finite as claimed.

By the same logic as in the remark following \eqref{entropOptimization}, we may assume that $Q_n \in \mathcal{Q}$, where $|\mathcal{Q}|=3$, since this is sufficient to preserve the values of $\EE[X_n^2]$, $\mathsf{s}_{\lambda}(X_n,\snr|Q_n)$ and $\left( h(Y_n|Q_n) - h(X_n|Q_n)\right)$.  Thus, since $\mathcal{Q}$ is finite and $\EE[X_n^2]\leq1$, the sequence $\{X_n,Q_n\}$ is tight.  By Prokhorov's theorem \cite{durrett2010probability}, we may assume that there is some $(X_*,Q_*)$ for which $(X_n,Q_n) \xrightarrow{\mathcal{D}} (X_*,Q_*)$ by restricting our attention to a subsequence of $\{X_n, Q_n\}$ if necessary.  Moreover, $\EE[X_*^2] \leq  \liminf_{n\to\infty} \EE[X_n^2]\leq 1$ by Fatou's lemma.%

Next, for a given $n$, let $(X_{1,n},Q_{1,n})$ and $(X_{2,n},Q_{2,n})$ denote two independent copies of $(X_n,Q_n)$. 
 Define 
\begin{align}
&\xh{+,n} = \frac{X_{1,n} + X_{2,n}}{\sqrt{2}} &\xh{-,n} = \frac{X_{1,n} - X_{2,n}}{\sqrt{2}},\label{hatTransform}
\end{align}
In a similar manner, define $\yh{+,n}, \yh{-,n}$, and put $\mathbf{Q}_n=(Q_{1,n},Q_{2,n})$. Applying Lemma \ref{lem:Gaussiansubadditivity} to  the variables $\mathbf{Q}_n\to (\xh{+,n},\xh{-,n})\to (\yh{+,n},\yh{-,n})  $, we obtain
\begin{align}
2 \mathsf{s}_{\lambda}(X_n,\snr|Q_n) \geq \mathsf{s}_{\lambda}(X_{+,n},\snr| X_{-,n} \mathbf{Q}_n) + \mathsf{s}_{\lambda}(X_{-,n},\snr| Y_{+,n} \mathbf{Q}_n),  \label{lastIneq}
\end{align}
and the symmetric inequality 
\begin{align}
2 \mathsf{s}_{\lambda}(X_n,\snr|Q_n) \geq \mathsf{s}_{\lambda}(X_{+,n},\snr| Y_{-,n} \mathbf{Q}_n) + \mathsf{s}_{\lambda}(X_{-,n},\snr| X_{+,n} \mathbf{Q}_n).  \label{lastIneq2}
\end{align}
By independence of $X_{1,n}$ and $X_{2,n}$ and the assumption that $\EE[X_n]=0$, we have 
\begin{align}
\EE[X^2_{+,n}] = \EE[X^2_{-,n}] = \frac{1}{2}\EE[X^2_{1,n}] + \frac{1}{2}\EE[X^2_{2,n}]  = \EE[X_n^2]\leq 1.
\end{align}
Hence, it follows that the terms in the RHS of \eqref{lastIneq} and the RHS of \eqref{lastIneq2} are each lower bounded by $\mathsf{V}_{\lambda}(\snr)$.  Since $\lim_{n\to \infty} \mathsf{s}_{\lambda}(X_n,\snr|Q_n) = \mathsf{V}_{\lambda}(\snr)$ by definition, we must also have
\begin{align}
\lim_{n\to\infty} \frac{1}{2}\Big(  \mathsf{s}_{\lambda}(X_{+,n},\snr| Y_{-,n} \mathbf{Q}_n)  +   \mathsf{s}_{\lambda}(X_{-,n},\snr| Y_{+,n} \mathbf{Q}_n) \Big)  = \mathsf{V}_{\lambda}(\snr).
\end{align}
In particular, by letting the random pair $(X'_n,Q'_n)$  correspond to equal time-sharing between the  pairs $(X_{+,n},  (Y_{-,n} \mathbf{Q}_n))$ and $(X_{-,n},  (Y_{+,n} \mathbf{Q}_n))$,  
we have constructed an admissible  sequence $\{X'_n, Q'_n\}$ which satisfies 
\begin{align}
\lim_{n\to \infty} \mathsf{s}_{\lambda}(X'_n,\snr|Q'_n) = \mathsf{V}_{\lambda}(\snr). \label{primeOptimal}
\end{align}
Using Markovity, the following identity is readily established
\begin{align}
h(Y_n|Q_n) - h(X_n|Q_n) &=~\frac{1}{2}\left( h(\yh{+,n},\yh{-,n}|\mathbf{Q}_n) - h(\xh{+,n},\xh{-,n}|\mathbf{Q}_n)\right)\\
&= ~\frac{1}{2}\left( h(\yh{-,n}|\yh{+,n},\mathbf{Q}_n) - h(\xh{-,n}|\yh{+,n},\mathbf{Q}_n) \right) \\
&~~+\frac{1}{2}\left(h(\yh{+,n}|\yh{-,n},\mathbf{Q}_n) - h(\xh{+,n}|\yh{-,n},\mathbf{Q}_n)\right)   \notag\\
&~~+ \frac{1}{2}I(\xh{+,n};\xh{-,n}|\yh{+,n},\yh{-,n},\mathbf{Q}_n) \notag\\
&=h(Y'_n|Q'_n) - h(X'_n|Q'_n) + \frac{1}{2}I(\xh{+,n};\xh{-,n}|\yh{+,n},\yh{-,n},\mathbf{Q}_n). \label{sandwichMI}
\end{align}
Since the sequence $\{X'_n,Q'_n\}$ is admissible, it must also satisfy \eqref{smallestH}.  Therefore, in view of \eqref{sandwichMI} and the fact that the LHS of \eqref{smallestH} is finite, this implies that 
\begin{align}
\liminf_{n\to\infty} I(X_{1,n}+X_{2,n};X_{1,n}-X_{2,n}|Y_{1,n},Y_{2,n},\mathbf{Q}_n)=0.
\end{align}
In particular, for $P_{Q_*} \times P_{Q_*}$-a.e. $(q_1,q_2)$, 
\begin{align}
\liminf_{n\to\infty} I(X_{1,n}+X_{2,n};X_{1,n}-X_{2,n}|Y_{1,n},Y_{2,n},(\mathbf{Q}_n=q_1,q_2))=0.
\end{align}
This completes the proof since  Lemma \ref{lem:GaussianCharMI} guarantees that, for $P_{Q_*}$-a.e. $q$,  the random variable $X_*|\{Q_* = q\}$ is normal with variance not depending on $q$, and moreover we have already observed that $\EE[X_*^2] \leq 1$, so the variance of $X_*|\{Q_* = q\}$ is at most unity as claimed.  
\end{proof}

\subsection{Weak Semicontinuity  of $\mathsf{s}_{\lambda}(\cdot ,\snr)$ } \label{subsec:weakSemiContinuity}

This subsection is devoted to establishing the following semicontinuity property of $\mathsf{s}_{\lambda}(\cdot ,\snr)$, which was the second essential ingredient  needed for the proof of Theorem \ref{thm:explicitVsnr} (i.e., Claim II).  
\begin{lemma}\label{lem:weakContinuity}
If $X_n  \xrightarrow{\mathcal{D}} X_* \sim N(\mu,\sigma_X^2) $ and $\sup_n \EE[X_n^2] <\infty$, then
\begin{align}
\liminf_{n\to \infty} \mathsf{s}_{\lambda}(X_n,\snr)  \geq  \mathsf{s}_{\lambda}(X_*,\snr) .
\end{align}
\end{lemma}
Recall that $ \mathsf{s}_{\lambda}(X,\snr)$ is defined in terms of the Gaussian channel $Y = \sqrt{\snr } X + Z$.  However, for the purposes of the proof, it will be convenient to omit the $\snr$ scaling factor, and instead parametrize the channel in terms of the noise variance.  Toward this end,  let $Z\sim N(0,\sigma^2)$.  For $\lambda > 0 $ and a random variable $X\sim P_X$, independent of $Z$, define $Y=X+Z$ and the functionals
\begin{align}
\mathsf{F}_{\lambda,\sigma^2}(X) &= \inf_{V : X \to Y\to V} \Big(I(Y;V)-\lambda I(X;V)\Big)\\
\mathsf{G}_{\lambda,\sigma^2}(X) &=  -h(X) + \lambda h(Y).
\end{align}
Lemma \ref{lem:weakContinuity} is an immediate corollary of weak lower semicontinuity of $\mathsf{G}_{\lambda,\sigma^2}(X)$ and $\mathsf{F}_{\lambda,\sigma^2}(X)$ at Gaussian $X$.  These facts are established separately below in Lemmas \ref{lem:weakContinuityG} and \ref{lem:weakContinuityF}, respectively.  The former is straightforward, while the latter requires some effort.

\begin{lemma}\label{lem:weakContinuityG}
If $X_n  \xrightarrow{\mathcal{D}} X_* \sim N(\mu,\sigma_X^2) $ and $\sup_n \EE[X_n^2] <\infty$, then
\begin{align}
\liminf_{n\to \infty} \mathsf{G}_{\lambda,\sigma^2}(X_n) \geq \mathsf{G}_{\lambda,\sigma^2}(X_*).
\end{align}
\end{lemma}
\begin{proof}
Fix $\delta>0$ and define $N_{\delta}\sim N(0,\delta)$, pairwise independent of $\{X_n\},X_*$. Observe that
\begin{align}
\mathsf{G}_{\lambda,\sigma^2}(X_n) &=  -h(X_n) + \lambda h(Y_n) \geq -h(X_n+N_{\delta}) + \lambda h(Y_n).
\end{align}  
By the third claim of Lemma \ref{lem:weakConvergenceGaussian}, we have $-h(X_n+N_{\delta}) + \lambda h(Y_n) \to -h(X_*+N_{\delta})+ \lambda h(Y_*)$ as $n\to \infty$.  Thus, 
\begin{align}
\liminf_{n\to\infty} \mathsf{G}_{\lambda,\sigma^2}(X_n) \geq -h(X_*+N_{\delta})+ \lambda h(Y_*).\label{deltaIneq}
\end{align}
Since $h(X_*+N_{\delta}) = \frac{1}{2}\log \left( 2 \pi e (\sigma_X^2 + \delta)\right)$ is continuous in $\delta$, we may take $\delta\downarrow 0$ to prove the claim.
\end{proof}

\begin{lemma}\label{lem:Fproperties}
 $\mathsf{F}_{\lambda,\sigma^2}(X)$ is continuous in $\lambda$.  Furthermore, if $X\sim N(\mu,\sigma_X^2)$, then 
\begin{align}
\mathsf{F}_{\lambda,\sigma^2}(X ) = \begin{cases}
\frac{1}{2}\left[\log\left((\lambda-1)\frac{\sigma_X^2}{\sigma^2}\right)-\lambda\log 
\left(\frac{\lambda-1}{\lambda}\left(1+ \frac{\sigma_X^2}{\sigma^2}\right)\right)\right] & \mbox{If $\lambda \geq 1 + \frac{\sigma^2}{\sigma_X^2}$}\\
0 & \mbox{If $0 \leq \lambda \leq 1 + \frac{\sigma^2}{\sigma_X^2}$.}
\end{cases} \label{explicitExpressionIB}
\end{align}
In particular, $\mathsf{F}_{\lambda,\sigma^2}(X)$ is continuous in the parameters $\sigma^2$, $\sigma_X^2$ and $\lambda$ for Gaussian $X$.
\end{lemma}
\begin{proof}
The function $\mathsf{F}_{\lambda,\sigma^2}(X)$ is the pointwise infimum of linear functions in $\lambda$, and is therefore concave and continuous on the open interval $\lambda\in(0,\infty)$ for any distribution $P_X$.   The explicit expression \eqref{explicitExpressionIB} follows by identifying $\gamma\, \snr \leftarrow \frac{\sigma_X^2}{\sigma^2}$ in Proposition \ref{prop:InfoBottleExplicit}.
\end{proof}

\begin{lemma}\label{lem:weakContinuityF}
If $X_n  \xrightarrow{\mathcal{D}} X_* \sim N(\mu,\sigma_X^2) $ and $\sup_n \EE[X_n^2] <\infty$, then
\begin{align}
\liminf_{n\to \infty} \mathsf{F}_{\lambda,\sigma^2}(X_n) \geq \mathsf{F}_{\lambda,\sigma^2}(X_*).
\end{align}
\end{lemma}
\begin{proof}
Fix an interval $B=[-b,b]$, a channel $P_{V|Y}$, and $\delta$ satisfying $0 < \delta < \sigma^2/2$.  Recalling the definition of $Z\sim N(0,\sigma^2)$, decompose $Z = N_1 + N_2 + N_3$, where $N_1\sim N(0,\delta)$, $N_2\sim N(0,\sigma^2-2\delta)$ and $N_3\sim N(0,\delta)$ are mutually independent.  Define $X_n^{\delta} = X_n + N_1$ and $Y_n^{\delta} = Y_n - N_3 = X_n + N_1+N_2$.  Note that we have  $X_n \to X_n^{\delta}\to Y_n^{\delta}\to Y_n\to V_n$, where $V_n$ is defined by the stochastic transformation $P_{V|Y}: Y_n \mapsto V_n$.  Using the notation of  Lemma \ref{lem:weakConvergenceGaussian}, we also have $X_* \to X_*^{\delta}\to Y_*^{\delta}\to Y_*\to V_*$, where $Y_*= X_* +Z$, $X_*^{\delta} = X_*+N_1$, $Y_*^{\delta} = Y_*-N_3$ and $V_*$ is defined via  $P_{V|Y}: Y_* \mapsto V_*$.  With these definitions in hand, we may apply Lemma \ref{lem:MIbounds} to the processes $\{X^{\delta}_n\},\{Y^{\delta}_n\}$   to conclude the existence of a sequence $\epsilon_n\to 0$, not depending on $P_{V|Y}$, that satisfies
\begin{align}
I(V_n; X^{\delta}_n|X^{\delta}_n \in B) &\leq (1+\epsilon_n)^2  I(V_*;X^{\delta}_*|X^{\delta}_* \in B)  -  (1+\epsilon_n)^2 \log (1-\epsilon_n )^2\label{MIineq1a}\\
I(V_n; Y^{\delta}_n|Y^{\delta}_n \in B) &\geq (1-\epsilon_n)^2  I(V_*;Y^{\delta}_*|Y^{\delta}_* \in B)  -  (1-\epsilon_n)^2 \log (1+\epsilon_n )^2\label{MIineq2a}\\
\PP(X^{\delta}_n \in B) &\leq (1+\epsilon_n) \PP(X^{\delta}_* \in B)\label{eq:PXB}\\
\PP(Y^{\delta}_n \in B) &\geq (1-\epsilon_n) \PP(Y^{\delta}_* \in B).\label{eq:PYB}
\end{align}
Now, we have the following sequence of inequalities
\begin{align}
&I(Y_n;V_n) - \lambda I(X_n;V_n)  \notag\\
&\geq I(Y^{\delta}_n;V_n) - \lambda I(X^{\delta}_n;V_n) \label{eq:perturbDPI}\\
&= I(Y^{\delta}_n,\mathds{1}_{\{Y^{\delta}_n \in B \}};V_n) - \lambda I(X^{\delta}_n,\mathds{1}_{\{X^{\delta}_n \in B \}};V_n) \label{eq:funOf}\\
&=\PP(Y^{\delta}_n \in B) I(Y^{\delta}_n;V_n| Y^{\delta}_n \in B  ) + \PP(Y^{\delta}_n \notin B) I(Y^{\delta}_n;V_n| Y^{\delta}_n \notin B  )  + I( \mathds{1}_{\{Y^{\delta}_n \in B \}};V_n)\label{eq:chainRule}\\
&~~~-\lambda \Big( \PP(X^{\delta}_n \in B) I(X^{\delta}_n;V_n| X^{\delta}_n \in B  ) + \PP(X^{\delta}_n \notin B) I(X^{\delta}_n;V_n| X^{\delta}_n \notin B  )  + I( \mathds{1}_{\{X^{\delta}_n \in B \}};V_n) \Big)\notag\\
&\geq \PP(Y^{\delta}_n \in B) I(Y^{\delta}_n;V_n| Y^{\delta}_n \in B  )  \label{eq:nonNegMI_entrUB}\\
&~~~-\lambda \Big( \PP(X^{\delta}_n \in B) I(X^{\delta}_n;V_n| X^{\delta}_n \in B  ) + \PP(X^{\delta}_n \notin B) I(X^{\delta}_n;Y_n| X^{\delta}_n \notin B  )  + H( \mathds{1}_{\{X^{\delta}_n \in B \}}) \Big)\notag\\
&\geq  \PP(Y^{\delta}_n \in B)(1-\epsilon_n)^2  I(Y^{\delta}_*;V_*|Y^{\delta}_* \in B)  -  \PP(Y^{\delta}_n \in B)(1-\epsilon_n)^2 \log (1+\epsilon_n )^2\label{eq:lemClose}\\
&~~~-\lambda \Big( \PP(X^{\delta}_n \in B) (1+\epsilon_n)^2  I(X^{\delta}_*;V_*|X^{\delta}_* \in B)  -  \PP(X^{\delta}_n \in B) (1+\epsilon_n)^2 \log (1-\epsilon_n )^2      \Big) \notag\\
&~~~-\lambda \Big(   P(X^{\delta}_n \notin B) I(X^{\delta}_n;Y_n| X^{\delta}_n \notin B  )  + H( \mathds{1}_{\{X^{\delta}_n \in B \}} ) \Big)\notag\\
&\geq \frac{\PP(Y^{\delta}_n \in B)}{\PP(Y^{\delta}_* \in B)}(1-\epsilon_n)^2  \Big(I(Y^{\delta}_*;V_*) -  \PP(Y^{\delta}_* \notin B)  I(Y^{\delta}_*;V_*|Y^{\delta}_* \notin B)  - I( \mathds{1}_{\{Y^{\delta}_* \in B \}};V_*) \Big)  \label{eq:MIchain2}\\
&~~~-\lambda (1+\epsilon_n)^2 \frac{\PP(X^{\delta}_n \in B)}{\PP(X^{\delta}_* \in B)}    I(X^{\delta}_*;V_*) \notag \\
&~~~-  \PP(Y^{\delta}_n \in B)(1-\epsilon_n)^2 \log (1+\epsilon_n )^2 + \lambda  \PP(X^{\delta}_n \in B) (1+\epsilon_n)^2 \log (1-\epsilon_n )^2  \notag\\
&~~~-\lambda \Big(   \PP(X^{\delta}_n \notin B) I(X^{\delta}_n;Y_n| X^{\delta}_n \notin B  )  + H( \mathds{1}_{\{X^{\delta}_n \in B \}}) \Big) \notag \\
&\geq (1-\epsilon_n)^3  I(Y^{\delta}_*;V_*)  -\lambda  (1+\epsilon_n)^3  I(X^{\delta}_*;V_*)  \label{probUppBnd}\\
&~~~-(1-\epsilon_n)^3 \Big( \PP(Y^{\delta}_* \notin B)  I(Y^{\delta}_*;Y_*|Y^{\delta}_* \notin B) + H( \mathds{1}_{\{Y^{\delta}_* \in B \}}) \Big)  \notag   \\
&~~~-  \PP(Y^{\delta}_n \in B)(1-\epsilon_n)^2 \log (1+\epsilon_n )^2 + \lambda  \PP(X^{\delta}_n \in B) (1+\epsilon_n)^2 \log (1-\epsilon_n )^2 \notag \\
&~~~-\lambda \Big(   \PP(X^{\delta}_n \notin B) I(X^{\delta}_n;Y_n| X^{\delta}_n \notin B  )  + H( \mathds{1}_{\{X^{\delta}_n \in B \}}) \Big) \notag \\
&\geq (1-\epsilon_n)^3  \mathsf{F}_{\lambda_n,(\sigma^2-2\delta) }(X^{\delta}_*) \label{eq:finalF}\\
&~~~-(1-\epsilon_n)^3 \Big( \PP(Y^{\delta}_* \notin B)  I(Y^{\delta}_*;Y_*|Y^{\delta}_* \notin B)  + H( \mathds{1}_{\{Y^{\delta}_* \in B \} }) \Big)  \notag   \\
&~~~-  \PP(Y^{\delta}_n \in B)(1-\epsilon_n)^2 \log (1+\epsilon_n )^2 + \lambda  \PP(X^{\delta}_n \in B) (1+\epsilon_n)^2 \log (1-\epsilon_n )^2 \notag \\
&~~~-\lambda \Big(   \PP(X^{\delta}_n \notin B) I(X^{\delta}_n;Y_n| X^{\delta}_n \notin B  )  + H( \mathds{1}_{\{X^{\delta}_n \in B \}}) \Big),\notag
\end{align}
where $\lambda_n := \lambda \left( \frac{1+\epsilon_n}{1-\epsilon_n}\right)^3$.   The above steps are justified as follows:
\begin{itemize}
\item \eqref{eq:perturbDPI} follows by the data processing inequality.
\item \eqref{eq:funOf} follows since $\mathds{1}_{\{Y^{\delta}_n \in B \}}$ and $\mathds{1}_{\{X^{\delta}_n \in B \}}$ are functions of $Y^{\delta}_n$ and $X_n^{\delta}$, respectively.
\item \eqref{eq:chainRule} follows from the chain rule for mutual information.
\item \eqref{eq:nonNegMI_entrUB} follows from non-negativity of mutual information, the fact that $I( \mathds{1}_{\{X^{\delta}_n \in B \}};V_n) \leq H( \mathds{1}_{\{X^{\delta}_n \in B \}}) $, and the data processing inequality which implies $I(X^{\delta}_n;V_n| X^{\delta}_n \notin B  ) \leq I(X^{\delta}_n;Y_n| X^{\delta}_n \notin B  ) $.
\item \eqref{eq:lemClose} follows from \eqref{MIineq1a} and \eqref{MIineq2a}.
\item \eqref{eq:MIchain2} follows from the chain rule for mutual information, which implies 
\begin{align}
 I(Y^{\delta}_*;V_*|Y^{\delta}_* \in B)  =\frac{1}{ \PP(Y^{\delta}_* \in B) } \Big(I(Y^{\delta}_*;V_*) -  \PP(Y^{\delta}_* \notin B)  I(Y^{\delta}_*;V_*|Y^{\delta}_* \notin B)  - I( \mathds{1}_{\{Y^{\delta}_* \in B \}};V_*) \Big)
\end{align}
and, combined with non-negativity of mutual information, 
\begin{align}
 I(X^{\delta}_*;V_*|X^{\delta}_* \in B)  \leq \frac{1}{\PP(X^{\delta}_* \in B)}   I(X^{\delta}_*;V_*).
\end{align}

\item \eqref{probUppBnd} follows from \eqref{eq:PXB}, \eqref{eq:PYB}, the fact that $I( \mathds{1}_{\{Y^{\delta}_* \in B \}};V_n) \leq H( \mathds{1}_{\{Y^{\delta}_* \in B \}})$, and the data processing inequality which implies $I(Y^{\delta}_*;V_*| Y^{\delta}_* \notin B  ) \leq I(Y^{\delta}_*;Y_*| Y^{\delta}_* \notin B  )$.  %

\item \eqref{eq:finalF} follows from the definition of $\mathsf{F}_{\lambda_n,(\sigma^2-2\delta) }(X^{\delta}_*)$ by taking the infimum over $V_*$ satisfying $X_*^{\delta}\to Y_*^{\delta}\to V_*$.  
\end{itemize}

Summarizing above, we have shown
\begin{align}
I(Y_n;V_n) - \lambda I(X_n;V_n)  &\geq (1-\epsilon_n)^3  \mathsf{F}_{\lambda_n,(\sigma^2-2\delta) }(X^{\delta}_*) \label{eq:finalF2}\\
&~~~-(1-\epsilon_n)^3 \Big( \PP(Y^{\delta}_* \notin B)  I(Y^{\delta}_*;Y_*|Y^{\delta}_* \notin B)  + H( \mathds{1}_{\{Y^{\delta}_* \in B \} }) \Big)  \notag   \\
&~~~-  \PP(Y^{\delta}_n \in B)(1-\epsilon_n)^2 \log (1+\epsilon_n )^2 + \lambda  \PP(X^{\delta}_n \in B) (1+\epsilon_n)^2 \log (1-\epsilon_n )^2 \notag \\
&~~~-\lambda \Big(   \PP(X^{\delta}_n \notin B) I(X^{\delta}_n;Y_n| X^{\delta}_n \notin B  )  + H( \mathds{1}_{\{X^{\delta}_n \in B \}}) \Big),\notag
\end{align}
Note that the RHS of \eqref{eq:finalF2} does not depend on $V_n$ (i.e., $P_{V|Y}$).  Thus, taking the infimum over $V_n$ satisfying $X_n \to Y_n\to V_n$   and then letting $n\to \infty$, we arrive at 
\begin{align}
 \liminf_{n\to\infty} \mathsf{F}_{\lambda ,\sigma^2}(X_n) &\geq \mathsf{F}_{\lambda,(\sigma^2-2\delta) }(X^{\delta}_*)  - \Big( \PP(Y^{\delta}_* \notin B)  I(Y^{\delta}_*;Y_*|Y^{\delta}_* \notin B)  + H( \mathds{1}_{\{Y^{\delta}_* \in B \}} ) \Big)  \label{finalF3}   \\
&~~~-\lambda \Big(   \PP(X^{\delta}_* \notin B) I(X^{\delta}_*;Y_*| X^{\delta}_* \notin B  )  + H( \mathds{1}_{\{X^{\delta}_* \in B \}}) \Big),\notag
\end{align}
which follows due to $\epsilon_n\to 0$ and the following:
\begin{itemize}
\item $\mathsf{F}_{\lambda_n,(\sigma^2-2\delta) }(X^{\delta}_*) \to  \mathsf{F}_{\lambda,(\sigma^2-2\delta) }(X^{\delta}_*)$ by continuity of $F_{\lambda,\sigma^2}(X)$ in $\lambda$ (Lemma \ref{lem:Fproperties}).
\item  $\PP(X^{\delta}_n \notin B) \to \PP(X^{\delta}_* \notin B)$ since $X^{\delta}_n\xrightarrow{\mathcal{D}} X^{\delta}_*$ by the first claim of Lemma \ref{lem:weakConvergenceGaussian}.  By the same token, $ H( \mathds{1}_{\{X^{\delta}_n \in B \}})\to  H( \mathds{1}_{\{X^{\delta}_* \in B \}})$ by continuity of the binary entropy function.
\item $I(X^{\delta}_n;Y_n| X^{\delta}_n \notin B  ) \to I(X^{\delta}_*;Y_*| X^{\delta}_* \notin B  ) $ by the third claim of Lemma \ref{lem:weakConvergenceGaussian} since $\limsup_n \EE[\left(X_n^{\delta}\right)^2|X^{\delta}_n \notin B ]<\infty$ due to the fact that $\sup_n \EE[X_n^2]<\infty$ and $\PP(X^{\delta}_n \notin B) \to \PP(X^{\delta}_* \notin B)$, a  positive constant.
\end{itemize}
As we take $b\to \infty$, continuity of the binary entropy function and Lemma \ref{lem:MItailBound} together imply the latter two terms in the RHS of \eqref{finalF3} vanish, yielding the  inequality
\begin{align}
 \liminf_{n\to\infty} \mathsf{F}_{\lambda }(X_n) &\geq \mathsf{F}_{\lambda,(\sigma^2-2\delta) }(X^{\delta}_*). \label{neededSemi}
\end{align}
Since $\delta$ was arbitrary and $\mathsf{F}_{\lambda,(\sigma^2-2\delta) }(X^{\delta}_*)$ is continuous in $\delta$ by Lemma \ref{lem:Fproperties}, the proof is complete by letting $\delta\downarrow 0$.
\end{proof}
\begin{remark}
Given the tedious chain of inequalities in the proof of Lemma \ref{lem:weakContinuityF}, it is easy to lose sight of the overall picture.  The crucial idea is that perturbing $X_n\to X_n^{\delta}$ and $Y_n\to Y_n^{\delta}$ allows us to eventually eliminate dependence on the channel $P_{V|Y}$ in the RHS of \eqref{eq:finalF}.  Resisting the temptation to take   limits $n\to \infty$ or $b\to \infty$ until after dependence on any particular channel $P_{V|Y}$ is eliminated (i.e., inequality \eqref{eq:finalF2}) is also essential.
\end{remark}

We note that the hypothesis that $X_*\sim N(0,\sigma^2)$ was not needed in the proof of Lemma \ref{lem:weakContinuityF} until the very last step.  Indeed, we may actually conclude  that the following general result holds, which may be of independent interest:
\begin{proposition}
Suppose $X_n  \xrightarrow{\mathcal{D}} X_*$ and $\sup_n \EE[X_n^2] <\infty$, then for all $0 < \delta < \delta' < \sigma^2$, the following holds:
\begin{align}
\liminf_{n\to \infty} \mathsf{F}_{\lambda,\sigma^2}(X_n) \geq \mathsf{F}_{\lambda,(\sigma^2-\delta') }(X_*+N_{\delta})  %
\end{align}
where $N_{\delta}\sim N(0,\delta)$ is independent of $X_*$.
\end{proposition}
\begin{proof}
The claim  follows from the proof of Lemma \ref{lem:weakContinuityF}, but stopping at \eqref{neededSemi} and not particularizing to Gaussian $X_*$. The replacement of $2\delta$ by $\delta'$ is straightforward by decomposing $Z$ differently in the first step of the proof. 
\end{proof}
\begin{remark}
It is  possible to establish weak upper semicontinuity of $\mathsf{F}_{\lambda,\sigma^2}(\cdot)$, but that is not needed for our purposes.
\end{remark}

\bibliographystyle{unsrt}
\bibliography{vectorGaussian}

\end{document}